\documentclass{IEEEojcsys}

\usepackage[colorlinks,urlcolor=blue,linkcolor=blue,citecolor=blue]{hyperref}

\usepackage{color,array}

\usepackage{graphicx}
\usepackage{amssymb}
\usepackage{bbm}
\usepackage{amsmath}
\allowdisplaybreaks
\usepackage{dsfont}
\usepackage{subfig}
\usepackage{algorithm}
\usepackage{algpseudocode}
\usepackage{xcolor}

\newtheorem{myth}{Theorem}
\newtheorem{lemma}{Lemma}
\newtheorem{remark}{Remark}
\newtheorem{assumption}{Assumption}

\newtheorem{corollary}{Corollary}[myth]

\newcommand{\wpp}{\textit{w}.\textit{p}.}
\newcommand{\ie}{\textit{i}.\textit{e}.}

\newcommand{\modified}[2]{\label{rev:Q#1}{\color{black}#2}}

\begin{document}

\sptitle{Article Category}

\title{Model-free LQG Control with Chance Constraints} 


\author{A. Naha\affilmark{1} (Senior Member, IEEE)}

\author{S. Dey\affilmark{2}  (Fellow, IEEE)}

\affil{Department of Electrical Engineering, Linköping University, 58183 Linköping, Sweden.} 
\affil{Department of Electrical Engineering, Uppsala University, 75103 Uppsala, Sweden.} 

\corresp{CORRESPONDING AUTHOR: A. Naha (e-mail: \href{mailto:arunava.naha@liu.se}{arunava.naha@liu.se})}
\authornote{This work was partially supported by the Swedish Research Council under grants 2017-04053 and the Wallenberg AI, Autonomous
Systems and Software Program (WASP) funded by the Knut and Alice Wallenberg Foundation. Computational resources were provided by the National Academic Infrastructure for Supercomputing in Sweden (NAISS), funded by the Swedish Research Council. \vspace{-10mm}}

\markboth{PREPARATION OF PAPERS FOR IEEE OPEN JOURNAL OF CONTROL SYSTEMS}{A. NAHA {\itshape ET AL}.}
\vspace{-6mm}
\begin{abstract}
This paper studies model-free optimal control design and its convergence properties for linear time-invariant systems subject to probabilistic risk or chance constraints. In particular, we study a natural policy gradient (NPG)–based actor–critic (AC) algorithm with two timescales, using a Lagrangian primal–dual framework to enforce the constraint. Furthermore, the risk is defined as the probability that a function of the one-step-ahead state exceeds a user-specified threshold. To our knowledge, this is the first work to study the analytical convergence properties for NPG–based AC in a chance-constrained linear–quadratic Gaussian (LQG) regulator  setting without model knowledge. We establish the coercivity and gradient dominance properties of the Lagrangian function, which ensure linear convergence and closed-loop stability during training for the actor. On the other hand, we analyse the convergence properties of the temporal difference (TD(0)) learning for the critic, applying stochastic approximation theory. Also, we demonstrate no duality gap in the constrained optimisation problem. Additionally, we have performed numerical analysis of the convergence properties and accuracy of the proposed method, comparing it with model-based chance-constrained LQR and scenario-based MPC. Results show that our approach effectively limits risk while maintaining near-optimal performance, without requiring full model knowledge or real-time optimisation.
\end{abstract}

\vspace{-6mm}
\maketitle

\section{INTRODUCTION} \label{sec:intro}
  The linear–quadratic regulator (LQR) and its stochastic extension, the linear–quadratic Gaussian (LQG) controller, are well studied for optimal control of linear systems with quadratic cost. Throughout this paper, we use the terms LQR and LQG interchangeably. However, these classical formulations are inherently risk-neutral and focus solely on minimising expected cost without explicitly accounting for the probability of rare but safety-critical events. Such risky or undesirable events may arise due to the long tail of process noise or disturbances. In many practical control applications, it is important to avoid such events proactively. For instance, an unmanned aerial vehicle (UAV) may need to steer clear of specific areas to remain undetected by adversaries. Therefore, it is crucial to design a controller that not only minimises the average expected control cost but also actively limits the likelihood of such high-impact occurrences \cite{tsiamisRiskConstrainedLinearQuadraticRegulators2020}. \modified{1.1}{A natural way to guard against such rare but high-impact events is to impose
chance (probabilistic risk) constraints, which explicitly bound the
probability that a function of the system state exceeds a critical threshold
\cite{tsiamisRiskConstrainedLinearQuadraticRegulators2020}.} In wind turbine control, for example, fluctuating wind speeds introduce uncertainty, and the control objective is to maximise power output while reducing the risk of structural damage \cite{schildbachScenarioApproachStochastic2014}. Likewise, in climate-controlled buildings, the goal is to minimise energy consumption while maintaining occupant comfort by keeping the temperature within a certain threshold \cite{flemingStochasticMPCAdditive2019}. As studied in \cite{flemingStochasticMPCAdditive2019}, controllers employing hard constraints tend to be overly conservative compared to those designed with softer, probabilistic constraints. In other words, the designed controller will lower the control cost if we constrain the probability of risky or undesirable events instead of imposing hard constraints. For example, in the wind turbine scenario, constraining the probability that blade stress exceeds a specified threshold can allow greater power generation than imposing a strict stress limit. 

Model predictive control (MPC) with chance constraints offers one solution, but it requires system models and online optimisation at each time step, leading to high computational costs and potential performance degradation under model mismatch. In contrast, model-free reinforcement learning (RL) methods, particularly policy gradient (PG)–based actor–critic (AC) algorithms, provide a flexible framework for directly learning optimal policies from data without explicit model identification. However, while RL methods have been successfully applied to unconstrained LQR problems, the theoretical understanding of their behaviour in chance-constrained settings remains limited. Moreover, the probabilistic constraints studied in much of the existing literature often involve indirect surrogates, such as variance bounds or expectations of penalty functions. Such approaches only provide loose guarantees on actual violation probabilities. Directly constraining the probability of risk events, though more meaningful, introduces significant analytical challenges, especially in the absence of model knowledge.
\subsection{Related work} \label{subsec:related_work}
The study of risk-aware optimal control spans several research communities, including MPC, stochastic control, and reinforcement learning. Here, we review a few relevant works in these areas. Additionally, we also focus on the convergence properties of the AC algorithms in this context.
\subsubsection{Chance-constrained-based approaches} In MPC, probabilistic risk is frequently modelled using chance constraints. A common approach is scenario-based sampling, where disturbance realisations are drawn from a known distribution, and the chance constraints are reformulated as deterministic algebraic constraints \cite{flemingStochasticMPCAdditive2019, schildbachScenarioApproachStochastic2014}. Alternative strategies replace chance constraints with tractable approximations, such as {moment-based bounds} (e.g., Chebyshev's inequality) or {expected-value formulations}, sometimes evaluated using {Hamiltonian Monte Carlo} (HMC) methods. For example, \cite{schildbachLinearControllerDesign2015} reformulates the chance-constrained LQR problem into a convex program with {linear matrix inequality} (LMI) constraints, which is then solved via {semi-definite programming} (SDP). Another important direction is {tube-based chance constraints}, ensuring that state and input trajectories remain inside invariant tubes \cite{arcariStochasticMPCRobustness2023, kerzDataDrivenTubeBasedStochastic2023}.
\subsubsection{Risk-Aware Infinite-Horizon Control}
Outside the MPC framework, risk has been incorporated into infinite-horizon formulations by constraining the average variance of a quadratic function of the state \cite{zhaoGlobalConvergencePolicy2022}. In these settings, the optimal controller can often be shown to be affine in the state \cite{zhaoInfinitehorizonRiskconstrainedLinear2021, tsiamisRiskConstrainedLinearQuadraticRegulators2020}. Such variance-based risk models offer analytical tractability but may only indirectly control the actual probability of constraint violation.
\subsubsection{Reinforcement Learning for LQR and Constrained Control}
Reinforcement learning (RL) methods have gained prominence for optimal control when system dynamics are partially or completely unknown \cite{bertsekasReinforcementLearningOptimal2019, busoniuReinforcementLearningControl2018, lopezEfficientOffPolicyQLearning2023}. Among these, policy gradient (PG)--based actor--critic (AC) algorithms are particularly suitable for continuous state--action problems \cite{suttonReinforcementLearningSecond2018}. For the standard LQR problem, PG-based methods have been analysed for global convergence and closed-loop stability, even though the optimisation landscape is non-convex in the control gain \cite{fazelGlobalConvergencePolicy2018a, huTheoreticalFoundationPolicy2023, yangProvablyGlobalConvergence2019,chen2025global}.

Extending PG analysis to {constrained} LQR settings remains challenging and has so far been addressed only in a limited number of special cases \cite{zhangPolicyOptimizationMathcal2021a,zhaoGlobalConvergencePolicy2023a,zhao2025policy,hanReinforcementLearningControl2021}. These methods are discussed in the subsequent subsection on convergence studies. A related line of work enforces safety in RL through control barrier functions (CBFs), which translate state-based safety requirements into pointwise forward-invariance constraints on the control input \cite{guerrier2024learning}. These methods integrate CBF-based safety filters with model-free policy gradient algorithms using Gaussian-process dynamics models \cite{cheng2019end}, or couple closed-form CBF filters with barrier-inspired reward shaping during training \cite{yang2025cbf}. These approaches provide strong forward-invariance guarantees, but do not directly bound the long-run average probability of constraint violation under stochastic disturbances, which is the focus of our chance-constrained LQR formulation.

\subsubsection{RL for Probabilistic Risk-Constrained LQR}
More recently, \cite{naha2023reinforcement} proposed a deep deterministic policy gradient (DDPG)--based AC approach for the probabilistic risk-constrained LQR problem, modelling risk as the probability that a quadratic function of the state exceeds a threshold. While empirical results are encouraging, a general theoretical understanding of PG-based AC algorithms in the probabilistic chance-constrained LQR setting, particularly in the model-free case, remains largely unaddressed, providing the motivation for this work.

\subsubsection{Convergence Studies}
The convergence properties of PG-based AC algorithms in the context of chance-constrained LQR problems are still an open area of research. Existing studies have primarily focused on unconstrained settings, leaving a gap in understanding how these algorithms behave when faced with probabilistic constraints. Establishing convergence guarantees in this more complex setting is crucial for ensuring the reliability and safety of the learned policies. Even in the risk-neutral model-free LQR case, convergence analysis is challenging due to the nonconvexity of the cost with respect to policy parameters. Moreover, in closed-loop applications, it is essential to maintain system stability while learning an online policy. Previous works on the standard LQR problem have analysed PG-based RL algorithms for global convergence and closed-loop stability \cite{fazelGlobalConvergencePolicy2018a, huTheoreticalFoundationPolicy2023, yangProvablyGlobalConvergence2019, zhou2023single}. These analyses typically exploit the gradient-dominance property of the cost function with respect to the policy parameters, which is instrumental in proving convergence. They also leverage the ergodicity of the cost function to guarantee stability of the closed-loop system during learning.

The analysis of closed-loop stability and convergence for PG-based algorithms in LQR problems with additional constraints is notably challenging and has so far been addressed only in a limited number of special cases. For instance, \cite{zhangPolicyOptimizationMathcal2021a} considers an additional $H_\infty$ robustness constraint. In \cite{zhaoGlobalConvergencePolicy2023a}, the authors establish the global convergence of a PG algorithm for a risk-constrained LQR, where risk is defined as the average variance of a quadratic function of the states over an infinite horizon. In contrast, \cite{hanReinforcementLearningControl2021} studies a safety constraint expressed as the expected value of a continuous non-negative state function being bounded by a prescribed threshold, deriving the optimal controller using an AC algorithm within a Markov decision process (MDP) framework.

The convergence properties of AC algorithms have also been studied in problem settings beyond LQR. However, most of these studies adopt a discount–factor–based reward formulation and often assume a finite action space \cite{wu2020finite, tian2023convergence}, making their results not directly applicable to our setting. Under a discount-based formulation, the cost function loses its coercivity property, which in turn makes maintaining closed-loop stability during online policy learning more challenging \cite{tangAnalysisOptimizationLandscape2021,huTheoreticalFoundationPolicy2023}.

For PG-based methods, convergence analysis for the actor typically focuses on establishing the gradient-dominance property of the cost function, while the convergence of the critic is often studied using stochastic approximation theory \cite{borkar2008stochastic}. When the critic is approximated as a linear function of features extracted from the states, convergence results have been established in several works \cite{mitra2024simple, dalal2018finite}. More recently, a number of studies have investigated the convergence of critics implemented as neural network function approximators \cite{cai2019neural, tian2023convergence}, primarily under overparameterisation assumptions and within the discounted-cost framework.

\subsection{Our Approach and Contributions} \label{subsec:contributions} 
\modified{1.2}{In this work, we study an online natural policy gradient (NPG)--based AC
algorithm
\cite{rajeswaranGeneralizationSimplicityContinuous2017,kakadeNaturalPolicyGradient2001}
for the chance-constrained LQG problem, assuming the model parameters are
unknown. The optimal policy is
designed within the class of linear state-feedback controllers, a
Lagrangian-based primal-dual formulation handles the chance constraint, and
the critic is a linear function of features extracted from the states,
enabling a systematic convergence analysis. We model risk as the probability
that a function of the one-step-ahead state exceeds a user-specified
threshold, and we constrain its long-term average to remain within a
prescribed bound. Unlike indirect formulations based on the average variance
of a quadratic state function \cite{zhaoGlobalConvergencePolicy2023a} or the
expected value of a non-negative state function
\cite{hanReinforcementLearningControl2021}, which provide only loose bounds on
the true violation probability. On the other hand, our formulation bounds this probability
directly, which is a more interpretable and physically meaningful constraint. However, such a formulation rarely admits a closed-form representation, making the convergence
analysis considerably more challenging. In the unknown-model setting, the
probability term in the reward is replaced by an indicator function, enabling
direct evaluation from observed data.


For the proposed two-time-scale AC algorithm, we first establish, via
stochastic approximation theory under the average expected cost formulation,
that the critic parameters converge with high probability to a neighbourhood of
their optimal values within a finite number of update steps. Exploiting the
gradient-dominance property, we then show that the actor parameters, updated
using the approximate critic, converge with high probability to a neighbourhood
of the optimum, while the coercivity property ensures closed-loop stability
during online learning.

A preliminary version of this work appeared in our conference paper
\cite{naha2025policy}, which used trajectory-based training and established
actor convergence under a known system model. The present work removes this
assumption, and we adopt a random-sampling-based actor training scheme with a
TD(0)-trained linear critic and analyse the convergence of the entire AC
algorithm. This is a non-trivial extension requiring new theoretical tools tailored to the probabilistic risk-constrained LQG setting. We note that an identification-then-control approach is also possible for LTI systems. However, a controller designed using estimated matrices $(\hat{\bf A},\hat{\bf B})$ may not satisfy the chance constraint for the true system, since the violation probability can be sensitive to model-estimation errors. Guaranteeing the constraint would therefore require accounting for the identification error, which is not covered by the known-model analysis in \cite{naha2025policy}. We summarise our main contributions as follows.}\newline 
1) To the best of our knowledge, this work is the first to provide a rigorous convergence analysis of an NPG-based control design for a probabilistic risk (chance)-constrained LQG problem in the unknown-model setting. \newline 
2) We present, for the first time, a convergence analysis of the TD(0) algorithm for training a linear critic in the average-cost, infinite-horizon setting with continuous state and action spaces. \modified{2.3}{This critic convergence result is of independent interest beyond the LQG setting since its proof relies only on
Assumption~\ref{assumption:bounded_noise}, ergodicity, and linear
independence of the features, and hence applies to average-reward TD(0)
learning in general continuous-state MDPs satisfying these conditions.} \newline 
3) We perform extensive simulations comparing the proposed method with model-based chance-constrained LQR and scenario-based MPC. Results show that our approach maintains the specified risk bounds with only a slight cost increase, achieves performance comparable to model-based methods without requiring full model knowledge, and avoids the high online computational load of MPC.

\subsection{Organisation} \label{subsec:organization}
The rest of the paper is organised as follows. In Section \ref{sec:problem}, we present the problem formulation and the reward structure for the chance-constrained LQG problem. In Section \ref{sec:ac_method}, we present the PG-based algorithms, while the convergence properties of the proposed algorithm are studied in Section~\ref{sec:main_results}. Section \ref{sec:results} presents the numerical performance results of the proposed algorithms, followed by concluding remarks in Section \ref{sec:conclusion}.

\subsection{Notations}
\label{subsec:notations} 
Some special notations are given in Table~\ref{tab:notations}. 
\begin{table}[h!]
	\begin{center}
		\caption{Notations}
		\label{tab:notations}
		\begin{tabular}{l|p{60mm}} 
			\hline \hline
			Symbol & Description \\
			\hline
			${\rm I\!R}^{n}$ & The set of $n\times 1$ real vectors \\
			${\rm I\!R}^{m\times n}$ & The set of $m\times n$ real matrices \\
			${X}^T$ & Transpose of matrix or vector ${X}$ \\
			$\mathcal{N}(\mu,{\bf \Sigma})$ & Gaussian distribution with mean $\mu$ and variance $\bf \Sigma$  \\
			$\bf \Sigma \ge 0$ or $>0$  & positive semi-definite or definite matrix, respectively \\
			$[\cdot]_{ij}$ & $i$-th row and $j$-th column element of a matrix \\
			$||\cdot||$ & Spectral norm of a matrix or Euclidean norm of a vector \\
			$||\cdot||_F$ & Frobenius norm of a matrix  \\
			$\text{tr}(\cdot)$ & Trace of a matrix \\
			$ \mathbb{E}\left[\cdot\right]$ and $\mathbb{P}\left\{\cdot \right\}$ & Expectation operator and Probability measure, respectively \\
            $\left\{\hat \cdot \right\}$ & Estimated or approximated value \\
			$\mathds{1}_{\left\{ condition \right\}}$ & Indicator function, 1 if condition is true, 0 otherwise \\
            $\sigma(\cdot)$ and $\rho(\cdot)$ & Singular value and eigenvalue, respectively \\
			$\tilde{O}(\cdot)$ & Hides problem dependent constants and poly-logarithmic terms \\
			\hline \hline 
		\end{tabular}
	\end{center} 
\end{table}

\section{PROBLEM FORMULATION} \label{sec:problem}
We consider the following linear time-invariant (LTI) discrete-time system:
\begin{equation}
{\bf x}_{k+1} = {\bf A}{\bf x}_{k} + {\bf B}{\bf u}_{k} + {\bf w}_{k},
\label{eqn:state_eqn}
\end{equation}
where ${\bf x}_{k} \in \mathbb{R}^{n}$ and ${\bf u}_{k} \in \mathbb{R}^{p}$ denote the state and control input at time $k$, respectively, \modified{2.1.1}{and ${\bf w}_{k} \in \mathbb{R}^{n}$ is an independent and identically distributed (iid) process noise with probability density $f_w(\mathbf{w})$.} The system matrices are ${\bf A} \in \mathbb{R}^{n \times n}$ and ${\bf B} \in \mathbb{R}^{n \times p}$.

We assume that all states are measured and that the pair $({\bf A}, {\bf B})$ is stabilisable. In the standard LQR problem, the objective is to minimise the following long-term average quadratic cost:
\begin{equation}
J = \lim_{T \to \infty} \frac{1}{T} \sum_{k=1}^T \mathbb{E} \left[ {\bf x}_k^\top {\bf Q} {\bf x}_k + {\bf u}_k^\top {\bf R} {\bf u}_k \right],
\label{eqn:cost_fun_J}
\end{equation}
where ${\bf Q} \in \mathbb{R}^{n \times n}$ and ${\bf R} \in \mathbb{R}^{p \times p}$ are positive definite weighting matrices. We also assume that $({\bf A}, {\bf Q}^{1/2})$ is detectable.

When the noise is zero-mean with bounded second moment, the optimal control law is a linear state-feedback of the form \cite{bertsekasDynamicProgrammingOptimal}:
\begin{equation}
{\bf u}_k^* = -{\bf K} {\bf x}_k, \quad
{\bf K} = \left( {\bf B}^\top {\bf S} {\bf B} + {\bf R} \right)^{-1} {\bf B}^\top {\bf S} {\bf A},
\label{eqn:opt_u}
\end{equation}
where ${\bf S}$ is the positive definite solution to the following algebraic Riccati equation:
\begin{equation}
{\bf S} = {\bf A}^\top {\bf S} {\bf A} + {\bf Q} - {\bf A}^\top {\bf S} {\bf B} \left( {\bf B}^\top {\bf S} {\bf B} + {\bf R} \right)^{-1} {\bf B}^\top {\bf S} {\bf A}.
\end{equation}
However, the standard cost formulation \eqref{eqn:cost_fun_J} is risk-neutral and does not explicitly account for rare but safety-critical events. To address this, we introduce an additional probabilistic risk constraint, leading to the following optimisation problem:
\begin{equation}
\begin{aligned}
& \underset{{\bf u} \in \mathcal{U}}{\text{minimise}}
& & J \
& \text{subject to}
& & \bar J_c \leq \delta,
\end{aligned}
\label{eqn:opt_prob}
\end{equation}
where $J$ is as in \eqref{eqn:cost_fun_J} and $\bar J_c$ is given as:
\begin{equation}
\bar J_c = \lim_{T \to \infty} \frac{1}{T} \sum_{k=1}^T \mathbb{E} \left[ \mathbb{P} \left( f_c({\bf x}_{k+1}) \ge \epsilon  \middle|  \Psi_k \right) \right].
\label{eqn:J_c}
\end{equation}
Here, $\epsilon$ is a user-specified violation threshold, $\delta > 0$ is the maximum allowable long-term average violation probability, and $\Psi_k \triangleq \left\{ {\bf x}_\ell, {\bf u}_\ell \mid 0 \le \ell \le k \right\}$ denotes the information set available at time $k$. The function $f_c(\cdot)$ defines the risk measure applied to the next-step state ${\bf x}_{k+1}$.
\begin{remark}
Since the constraint violation probability, \ie, $\mathbb{P} \left( f_c({\bf x}_{k+1}) \ge \epsilon \right)$, depends on the random information set $\Psi_k$, the expectation with respect to $\Psi_k$ is taken in the formulation. Moreover, we require that the long-term average probability of these events remains bounded over an infinite time horizon.
\end{remark}
\subsection{Reward Structure} \label{subsec:reward}
The constrained optimisation problem in (\ref{eqn:opt_prob}) can be reformulated as an unconstrained stochastic control problem using the Lagrangian multiplier $\bar \lambda > 0$ as follows: 
\begin{equation} 
	 \bar{\mathcal{L}} = J + \bar \lambda \left(\bar J_c - \delta \right) 
     = \lim_{T\to\infty} \frac{1}{T} \sum_{k=1}^T 
     \mathbb{E}\!\left[g\!\left( {\bf x}_k,{\bf u}_k\right) \right],
	\label{eqn:JL}
\end{equation}
where the per-stage cost $g(\cdot)$ is defined as
\begin{align}
	g\left( {\bf x}_k,{\bf u}_k\right) &= f\left( {\bf x}_k,{\bf u}_k\right) 
	+ \bar \lambda \left( h_p\left( {\bf x}_k,{\bf u}_k\right) - \delta \right), \label{eqn:gk} \\
	f\left( {\bf x}_k,{\bf u}_k\right) &= {\bf x}^T_k{\bf Q}{\bf x}_k + {\bf u}^T_k{\bf R}{\bf u}_k,  \label{eqn:fk} \\
	h_p\left( {\bf x}_k,{\bf u}_k\right) &= 
    \mathbb{P}\!\left\{ f_c\!\left({\bf x}_{k+1} \right)\! \ge \epsilon \,\middle|\, \Psi_{k} \right\}.  
    \label{eqn:hpk} 
\end{align}
Since the per-stage cost may generally contain an intractable probabilistic term, we adopt a more practical reward formulation for RL-based algorithms. Given that future states ${\bf x}_{k+1}$ are available through stored trajectories, the per stage reward is
\begin{align}
    &{\color{blue}r({\bf x}_k,{\bf u}_k)} = -f\left( {\bf x}_k,{\bf u}_k\right) 
    -\bar \lambda \left( h_r\left( {\bf x}_{k+1}\right) - \delta \right), \text{and} \label{eqn:rk} \\
    &h_r\left( {\bf x}_{k+1}\right) = \mathds{1}_{\left\{ f_c\left({\bf x}_{k+1} \right) \ge \epsilon \right\}}, \label{eqn:hrk}
\end{align}
where $\mathds{1}_{\{\cdot\}}$ denotes the indicator function. Note that the indicator function based reward formulation is only used for the numerical simulations, whereas the probability based formulation is used for the theoretical analysis.

In the next section, we introduce the NPG-based AC algorithm employed in our study.

\section{AC METHOD FOR MODEL-FREE LQG CONTROL WITH CHANCE CONSTRAINTS} \label{sec:ac_method}
The complete algorithm used in our study is given in Algorithm \ref{alg:ac_method}. We assume the policy $\pi_\theta\left(\cdot \mid {\bf x}_k \right)$ to be stochastic but stationary. Here, $\theta$ denotes the policy parameter. The objective is to find the optimal policy parameter $\theta^*$ that maximises the following expected return. 
\begin{equation}
\mathcal{R} = \lim_{T \to \infty} \frac{1}{T} \sum_{k=1}^T \mathbb{E} \left[ r({\bf x}_k,{\bf u}_k) \right].
\label{eqn:expected_return}
\end{equation} 
In the NPG-based methods, the inverse of the Fisher information matrix $\textbf{F}$ is used to compute the steepest ascent direction as $\textbf{F}^{-1}G$ \cite{kakadeNaturalPolicyGradient2001}. Here, $G$ is the gradient of the expected return $\mathcal{R}$ with respect to the policy parameters $\theta$. The gradient is estimated from the samples collected during the interaction with the environment using the policy gradient theorem as follows \cite{suttonReinforcementLearningSecond2018}: 
\begin{equation}
\hat G = \frac{1}{N} \sum_{k=1}^N {\hat Q}\left( {\bf x}_k, {\bf u}_k \right) \nabla_\theta \log \pi_\theta\left({\bf u}_k \mid {\bf x}_k \right).
\label{eqn:policy_gradient}
\end{equation}

Here, $\hat{Q}\left( {\bf x}_k, {\bf u}_k \right)$ is the estimated action-value (Q) function, which approximates the true Q-function $Q(\cdot,\cdot)$ defined in (\ref{eqn:Q_fun}) \cite{yangProvablyGlobalConvergence2019}. The Q-function measures the expected cumulative deviation from the long-run average reward $\mathcal{R}$, starting from state ${\bf x}_k$ with action ${\bf u}_k$ and following the policy thereafter. Subtracting $\mathcal{R}$ in (\ref{eqn:Q_fun}) is standard in the average-reward setting and ensures that the infinite sum converges \cite{yangProvablyGlobalConvergence2019}.
\begin{equation}
Q\left( {\bf x}, {\bf u} \right) =  \sum_{k \ge 1}\left( \mathbb{E} \left[r({\bf x}_k, {\bf u}_k)  \mid {\bf x}_1 = {\bf x},{\bf u}_1 = {\bf u}\right] - \mathcal{R} \right)  .
\label{eqn:Q_fun}
\end{equation}
\begin{remark}
Note that there are a few other variants of the policy gradient studied in the literature where the action-value function is replaced with the advantage function \cite{rajeswaranGeneralizationSimplicityContinuous2017,schulmanHighDimensionalContinuousControl2018}. However, we have used the action-value function in our study because it provides a direct connection between the policy gradient update and the underlying stochastic approximation framework, which facilitates a more transparent convergence analysis. Moreover, working with $Q(x,u)$ avoids the additional layer of complexity introduced by estimating the advantage values, thereby enabling us to highlight the key theoretical properties of the algorithm without obscuring them behind additional complexities. \modified{2.4}{Additionally, we adopt the NPG rather than the vanilla PG because, for
LQR-type objectives, the Fisher preconditioning yields a Gauss--Newton-like
update direction, which removes the state-covariance factor $\Sigma_K$ from
the gradient and leads to a linear convergence rate whose step-size condition
does not degrade with the conditioning of $\Sigma_K$, see
Section~\ref{subsec:actor_convergence}. Moreover, approximating the action-value function with a linear-in-parameters critic places the critic update exactly within the average-reward TD(0) stochastic-approximation framework, which enables the finite-sample analysis in Section~\ref{subsec:critic_convergence}.}
\end{remark}
The Fisher information matrix is also estimated from the samples as follows:
\begin{equation}
\hat{\textbf{F}} = \frac{1}{N} \sum_{k=1}^N \nabla_\theta \log \pi_\theta\left({\bf u}_k \mid {\bf x}_k \right) \nabla_\theta \log \pi_\theta\left({\bf u}_k \mid {\bf x}_k \right)^\top.
\label{eqn:fisher_information}
\end{equation}
Finally, the policy parameter is updated using the following update rule:
\begin{equation}
\theta_{j+1} = \theta_j + \alpha_j \hat{\textbf{F}}^{-1} \hat{G},
\label{eqn:policy_update}
\end{equation}
where $\alpha_j$ is the step size.

On the other hand, the critic part of the AC method used in this study is a parameterised function approximating the Q-function, denoted as $\hat Q_\phi({\bf x}, {\bf u})$, where $\phi$ represents the parameters of the critic network. The critic is trained using the temporal difference (TD(0)) learning method, which updates the parameters $\phi$ based on the following update rule. 
The TD error for the $i$-th sample in a mini-batch is given by
\begin{equation}
    \delta^{(i)} = r^{(i)} - \hat{\mathcal{R}}  +  \hat Q_\phi\left({\bf x}'^{(i)}, {\bf u}'^{(i)}\right) - \hat Q_\phi\left({\bf x}^{(i)}, {\bf u}^{(i)}\right),
\end{equation}
where $({\bf x}^{(i)}, {\bf u}^{(i)}, r^{(i)}, {\bf x}'^{(i)}, {\bf u}'^{(i)})$ is a transition sample. In other words, the input ${\bf u}'^{(i)}$ was applied to the system at state ${\bf x}^{(i)}$ and observed the next state ${\bf x}'^{(i)}$. \modified{1.4}{We have used the notation $r({\bf x},{\bf u})$ to denote the per-stage reward (\ref{eqn:rk}) evaluated at the state--action pair $({\bf x},{\bf u})$, and for the $i$-th transition sample we write $r^{(i)} = r({\bf x}^{(i)},{\bf u}^{(i)})$.} $\hat{\mathcal{R}}$ is the estimated average reward.

The mean-squared TD error over a mini-batch of size $L_c$ is then
\begin{equation}
    \bar{\delta}^2 = \frac{1}{L_c} \sum_{i=1}^{L_c} \frac{1}{2} {\delta^{(i)}}^2.
    \label{eqn:mean_td_error}
\end{equation}
The critic parameters are updated by minimising the mean-squared TD error as
\begin{equation}
    \phi_{l+1} = \phi_l - \alpha_l \nabla_\phi \bar{\delta}^2,
    \label{eqn:critic_update}
\end{equation}
where $\alpha_l > 0$ is the learning rate. Furthermore, the average reward is updated as
\begin{equation}
    \hat{\mathcal{R}}_{l+1} = \hat{\mathcal{R}}_l + c_{\alpha} \alpha_l \frac{1}{L_c} \sum_{i=1}^{L_c} \left( r^{(i)} - \hat{\mathcal{R}}_l \right).
    \label{eqn:average_reward_update}
\end{equation}
Here, $c_{\alpha} > 0$ is a constant.

\begin{algorithm}[t]
\caption{Two-Timescale NPG-based AC (NPG--AC)}
\label{alg:ac_method}
\begin{algorithmic}[1]
\State \textbf{Inputs:} critic steps per actor step $N_c$, critic mini-batch size $L_c$, actor mini-batch size $N$, step sizes $\{\alpha_\ell\}$ (critic), $\{\alpha_j\}$ (actor).
\State \textbf{Initialise:} policy parameter $\theta_0$, critic parameter $\phi_0$, replay buffer $\mathcal{B} \gets \emptyset$
\Loop \Comment{actor iterations $j=0,1,2,\ldots$}
    \State \textbf{(Data collection)} Interact with the environment using $\pi_{\theta_j}(\cdot \mid {\bf x})$; append transitions $\big({\bf x},{\bf u}, r, {\bf x}'\big)$ to $\mathcal{B}$.
    \For{$t=1,\ldots,N_c$} \textbf{(Critic update)}
        \State Sample a mini-batch $\mathcal{M}_c$ of size $L_c$ from $\mathcal{B}$.
        \State Compute the mean-squared TD error using (\ref{eqn:mean_td_error}).
        \State Update $\phi$ according to (\ref{eqn:critic_update}).
    \EndFor
    \State \textbf{(Actor update)}
    \State Sample an actor mini-batch $\mathcal{M}_a$ of size $N$ from $\mathcal{B}$.
    \State Estimate the policy gradient $\hat G$ using (\ref{eqn:policy_gradient}) with the current critic $\hat Q_\phi$.
    \State Estimate the empirical Fisher information matrix $\hat{\bf F}$ using (\ref{eqn:fisher_information}).
    \State Update the policy parameter $\theta$ according to (\ref{eqn:policy_update}).
    \State Delete old transitions from $\mathcal{B}$.
\EndLoop
\State \textbf{Output:} final policy parameter $\theta$ and critic parameter $\phi$.
\end{algorithmic}
\end{algorithm}

\section{ANALYTICAL RESULTS} \label{sec:main_results}
In this section, we study the convergence properties of the NPG-AC algorithm (Algorithm~\ref{alg:ac_method}) for the chance constraint LQG problem formulated in Section~\ref{sec:problem}. First, we perform a finite sample analysis of the critic update method and derive an error bound for $N_c$ critic update iterations. In other words, we derive the bound on $\mid \mid Q({\bf x},{\bf u}) - \hat{Q}_\phi({\bf x},{\bf u}) \mid \mid \forall {\bf x},{\bf u}$. \modified{3.3.1}{Next, we establish the coercivity and gradient dominance properties of a
surrogate Lagrangian $\mathcal{L}({\bf K},\lambda)$, obtained by replacing the
chance constraint with an equivalent constraint on the stationary state
covariance. Since the two constraints define the same feasible state due to the bijective nature of the tail probability of the Gaussian distribution $Q_g(a)$ when $a$ is positive, and the objective function remains the same, the surrogate and the original optimisation problem (\ref{eqn:JL}) are equivalent, and an optimal solution to the surrogate problem is also an optimal solution to (\ref{eqn:JL}). Additionally, we
show that there is no duality gap under the surrogate formulation.} Finally, we prove the linear convergence of the actor update under the bounded error in the critic's value function approximation. 

\modified{1.5.2}{Our core novel contributions in this section are (i) the finite-sample analysis of the average-reward TD(0) critic with the combined parameter vector $\Phi=[\hat{\mathcal R},\phi^T]^T$ (Theorem~\ref{thm:critic_convergence}) (ii) the coercivity and approximate gradient-dominance properties of the chance-constrained surrogate Lagrangian (Lemmas~\ref{lemma:coercivity} and \ref{lemma:gradient_dominance}) (iii) the convergence of the actor under the approximate NPG with critic approximation error (Theorem~\ref{th:monotonicity}) and (iv) the strong-duality result (Lemma~\ref{lemma:duality}). The smoothness property (Lemma~\ref{lemma:Lsmooth}) and the convergence rate under the true NPG (Lemma~\ref{lemma:conv_rate}) adapt existing results for the unconstrained LQR problem \cite{huTheoreticalFoundationPolicy2023,yangProvablyGlobalConvergence2019} to our setting.}

\subsection{Critic Convergence Analysis} \label{subsec:critic_convergence}
In this subsection, we analyse the convergence properties of the TD(0) learning method used for the critic update in the NPG-AC algorithm. For this analysis, we consider the following structure of the critic function, 
\begin{equation}
    \hat{Q}_\phi({\bf x},{\bf u}) = z({\bf x},{\bf u})^T \phi.
    \label{eqn:critic_structure}
\end{equation}
Here, $z({\bf x},{\bf u})$ is a feature vector extracted from the state-action pair $({\bf x},{\bf u})$, and $\phi$ is the trainable parameter vector. Several design choices exist for the feature function $z(\cdot,\cdot)$, including polynomial features, radial basis functions, and neural networks. For notational simplicity, we denote ${\bf z} = z({\bf x},{\bf u})$ and ${\bf z}' = z({\bf x}',{\bf u}')$.

Furthermore, for the convergence analysis, we take the mini-batch size $L_c = 1$ in each critic update step. Therefore, the critic update in (\ref{eqn:critic_update}) and (\ref{eqn:average_reward_update}) can be expressed as \modified{3.1}{
\begin{align}
\phi_{l+1} &= \phi_l + \alpha_l\left[r - \hat{\mathcal{R}}_l + {\bf z}'^T \phi_l -  {\bf z}^T \phi_l \right]{\bf z} \text{, and} \label{eqn:phi_update_single} \\
\hat{\mathcal{R}}_{l+1} &= \hat{\mathcal{R}}_l + c_{\alpha} \alpha_l \left( r - \hat{\mathcal{R}}_l \right), \label{eqn:R_update_single}
\end{align}}
where $({\bf x},{\bf u},r,{\bf x}',{\bf u}')$ is a random transition sample drawn from the replay buffer $\mathcal{B}$. \modified{3.1.2}{Note that we have used the standard semi-gradient TD(0) method in which the bootstrapped target $(r - \hat{\mathcal R}_l + {\bf z}'^T\phi_l)$ is treated as fixed when differentiating \cite{suttonReinforcementLearningSecond2018}.} Furthermore, for notational simplicity, we combine the two update equations (\ref{eqn:phi_update_single}) and (\ref{eqn:R_update_single}) into a single update equation as follows:
\begin{equation}
    \Phi_{l+1} = \Phi_l + \alpha_l \left(A\left({\bf X} \right)\Phi_l + b \left({\bf X} \right)\right).
    \label{eqn:combined_update}
\end{equation}
Here, $\Phi_l = [\hat{\mathcal{R}}_l, \phi_l^T]^T$ is the combined parameter vector. Additionally, ${\bf X} = [{\bf x}^T, {\bf x}'^T, {\bf z}^T]^T$, and the matrices $A({\bf X})$ and $b({\bf X})$ are given as
\begin{equation}
    A({\bf X}) = \begin{bmatrix} -c_{\alpha} & 0 \\ -{\bf z} & {\bf z}({\bf z}' - {\bf z})^T \end{bmatrix}, \quad b({\bf X}) = \begin{bmatrix} c_{\alpha} r \\ r {\bf z} \end{bmatrix}.
    \label{eqn:A_b_matrices}
\end{equation}
Note that the matrices $A({\bf X})$ and $b({\bf X})$ are random due to the randomness in the transition sample ${\bf X}$. To study the convergence properties of the update equation (\ref{eqn:combined_update}), we have used the stochastic approximation theory \cite{borkar2008stochastic,dalal2018finite}. First, the update equation (\ref{eqn:combined_update}) is expressed in the following standard form, 
\begin{equation}
    \Phi_{l+1} = \Phi_l + \alpha_l \left(h\left(\Phi_l \right) + M_{l+1}\right).
    \label{eqn:sa_form}
\end{equation}
Here, $h(\Phi_l) = \bar{b} - \bar{A} \Phi_l $, where $\bar{A}$ is given in (\ref{eqn:Abar_1}) and $\bar{b} = \mathbb{E}[b({\bf X})]$. 
\begin{equation}
    \bar{A} = \begin{bmatrix} c_{\alpha} & \bf{0} \\ \bf{0} & \mathbb{E}\left[{\bf z}({\bf z} - {\bf z}')^T\right] \end{bmatrix}.
    \label{eqn:Abar_1}
\end{equation}
Furthermore, $M_{l+1}$ is a martingale difference noise, which is given by 
\begin{equation}
    M_{l+1} = A({\bf X}) \Phi_l + b({\bf X}) - h\left(\Phi_l\right).
    \label{eqn:martingale_noise}
\end{equation}
It is known that for a positive definite matrix $\bar{A}$, the ordinary differential equation (ODE) $\dot{\Phi} = h(\Phi)$ has a unique globally asymptotically stable equilibrium point given by $\Phi^* = \bar{A}^{-1} \bar{b}$ \cite{borkar2008stochastic}. The positive definiteness of $\bar{A}$ can be guaranteed by the following Lemma~\ref{lemma:pd_Az}. 
Finally, the convergence properties of the critic update are established under Assumption \ref{assumption:bounded_noise} in Theorem~\ref{thm:critic_convergence}. 
\begin{assumption} \label{assumption:bounded_noise}
    We assume $\mid \mid z({\bf x},{\bf u}) \mid \mid \le \frac{1}{2}$ $\forall ({\bf x},{\bf u})$, and $\mid r({\bf x},{\bf u}) \mid \le 1$, $\forall ({\bf x},{\bf u})$.
\end{assumption}
\begin{remark}
\modified{1.4.2}{In practice, the norm bound on $z({\bf x},{\bf u})$ can be achieved by normalisation, where the feature vector is scaled by a constant $Z_c$, chosen with respect to the operating region of the state and input, such that $\|z({\bf x},{\bf u})\| \le \frac{1}{2}$ holds over that region, see Section~\ref{sec:results}. Additionally, one can select a feature mapping that is inherently bounded such as the Fourier feature vectors. Such boundedness assumptions on the features and rewards are standard in the finite-sample analysis of TD learning \cite{dalal2018finite}.}
\end{remark}
\begin{lemma}\label{lemma:pd_Az} Assume that under any intermediate stabilising policy the closed-loop Markov process is ergodic and admits a unique stationary distribution. Then, the matrix $\bar{A}$, as given in (\ref{eqn:Abar_1}),  will be positive definite almost surely, provided the feature vectors $z({\bf x},{\bf u})$ are linearly independent.
\end{lemma}
\begin{proof}
The proof of Lemma~\ref{lemma:pd_Az} is provided in Appendix~\ref{app:proof_lemma_pd_Az}.
\end{proof}
\begin{myth} \label{thm:critic_convergence}
 We assume Assumption \ref{assumption:bounded_noise} to hold. Furthermore, let the step sizes $\alpha_l = \frac{1}{l+1}$ be chosen. Then, for any arbitrarily small $\epsilon > 0$ and $\delta_c \in (0,1)$, there exists a function
    \begin{equation}
        \begin{aligned}
        &L_c(\epsilon,\delta_c) = \\
        &\tilde{O}\left(\max \left\{ \left[ \frac{1}{\epsilon}\right]^{1+\frac{1}{\rho}} \left[\ln\left(\frac{1}{\delta_c}\right)\right]^{1+\frac{1}{\rho}}, \left[ \frac{1}{\epsilon}\right]^{2} \left[\ln\left(\frac{1}{\delta_c}\right)\right]^{3} \right\}\right),
        \end{aligned}
    \end{equation}
    such that 
    \begin{equation}
        \mathbb{P}\left\{\mid \mid \Phi_{l} - \Phi^* \mid \mid \le \epsilon \right\}  \ge 1 - \delta_c, \forall l \ge L_c(\epsilon,\delta_c).
    \end{equation}
    Here, $\rho \in \left(0, \min_{i \in [d]} \Re(\rho_i) \right)$, where $\rho_i$ are the eigenvalues of the matrix $\bar{A}$, and $d$ is the dimension of the parameter vector $\Phi$. 
\end{myth}
\begin{proof}
The proof of Theorem~\ref{thm:critic_convergence} is provided in Appendix~\ref{app:proof_thm_critic_convergence}.
\end{proof}
    Theorem~\ref{thm:critic_convergence} provides a finite sample bound for the critic update in the NPG-AC algorithm. It states that after $L_c(\epsilon,\delta_c)$ iterations of the critic update, the parameter vector $\Phi_l$ will be within an $\epsilon$-neighbourhood of the optimal parameter vector $\Phi^*$ with a probability of at least $1 - \delta_c$. The function $L_c(\epsilon,\delta_c)$ indicates how many iterations are needed to achieve this level of accuracy and confidence. The bound depends on the desired accuracy $\epsilon$, the confidence level $\delta_c$, and the eigenvalues of the matrix $\bar{A}$. The term $\tilde{O}(\cdot)$ hides problem dependent constants and poly-logarithmic terms. 

\begin{remark}
        \modified{1.5.1}{
    Theorem~\ref{thm:critic_convergence} adapts the proof template of
    Theorem~3.5 in \cite{dalal2018finite} to our setting. The novelty lies in
    the setting and in the verification of the conditions under which the
    template applies, not in the stochastic-approximation machinery itself.
    Specifically, \cite{dalal2018finite} analyses TD(0) for an approximate
    value function under the discounted setup, whereas we analyse an
    approximate Q-function under the average-reward setup, in which the
    average-reward estimate $\hat{\mathcal R}$ must be learned jointly with
    the critic parameters. This leads to the combined parameter vector
    $\Phi = [\hat{\mathcal R}, \phi^T]^T$ and, consequently, to matrices
    $\bar{A}$, $\bar{b}$ and a martingale noise $M_{l+1}$ with a block
    structure that does not appear in \cite{dalal2018finite}. In particular,
    the positive definiteness of our $\bar{A}$ does not directly follow from \cite{dalal2018finite} and is established separately in Lemma~\ref{lemma:pd_Az}. Furthermore, the average-reward setup is not a matter of convenience. Such a formulation is required for
 the coercivity of the cost function (Lemma~\ref{lemma:coercivity}), which
 in turn guarantees closed-loop stability during learning. On the other hand,  in a discounted formulation, coercivity is compromised
    \cite{tangAnalysisOptimizationLandscape2021}.}
\end{remark}

\begin{remark}
We assume that the samples $({\bf x},{\bf u},r,{\bf x}',{\bf u}')$ used for the TD(0) update are iid samples, which is a standard assumption in RL literature \cite{suttonReinforcementLearningSecond2018,dalal2018finite}. In practice, the samples are generated from the system dynamics and are not iid. However, it is possible to make the samples approximately iid by using a sufficiently large replay buffer and randomly sampling from it. This assumption is commonly used in practice and has been shown to work well empirically \cite{mnih2015human,lillicrapContinuousControlDeep2016}.
\end{remark}


\begin{assumption} \label{assumption:approximation_error}
    We assume that $\mid Q({\bf x},{\bf u}) - \hat{Q}^*_\phi({\bf x},{\bf u}) \mid \le M_Q$ $\forall ({\bf x},{\bf u})$, where $M_Q > 0$ is a small constant. Here, $\hat{Q}^*({\bf x},{\bf u})$ denotes the approximate optimal Q-function, \ie, with parameter $\Phi^*$. The implication of this assumption is that the feature vector $z({\bf x},{\bf u})$ is chosen in such a way that the optimal Q-function will approximate the true Q-value with a small error. This assumption is standard in the RL literature \cite{suttonReinforcementLearningSecond2018}.
\end{assumption}
Finally, we state the following corollary to bound the error in the Q-function approximation after $N_c$ critic update iterations.
\begin{corollary} \label{cor:critic_error_bound}
    We suppose Assumptions \ref{assumption:bounded_noise} and \ref{assumption:approximation_error} to hold. Furthermore, let the step sizes $\alpha_l = \frac{1}{l+1}$ be chosen. Then, for $\epsilon > 0$ and $\delta_c \in (0,1)$, there exists a function $L_c(\epsilon,\delta_c)$ as given in Theorem~\ref{thm:critic_convergence}, such that
        \begin{align}
        &\mathbb{P}\left\{\mid Q({\bf x},{\bf u}) - \hat{Q}_{\phi_{N_c}}({\bf x},{\bf u}) \mid \le \frac{\epsilon}{2} + M_Q, \forall ({\bf x},{\bf u}) \right\} \ge 1 - \delta_c, \nonumber \\
        & \forall N_c \ge L_c(\epsilon,\delta_c).
        \end{align}
\end{corollary}
\begin{proof}
The proof of Corollary~\ref{cor:critic_error_bound} follows directly from Theorem~\ref{thm:critic_convergence} and the structure of the critic function (\ref{eqn:critic_structure}).
\end{proof}
\begin{remark}
    Corollary~\ref{cor:critic_error_bound} provides a high probability finite sample bound on the error in the Q-function approximation after $N_c$ iterations of the critic update. Here, $M_Q$ is the inherent approximation error due to the choice of feature vector $z({\bf x},{\bf u})$, and $\frac{\epsilon}{2}$ is the additional error that can be made arbitrarily small by increasing the number of critic update iterations $N_c$ with high probability.
\end{remark}

In the next subsection, we study the convergence properties of the actor update in the NPG-AC algorithm under the condition that the critic function can be approximated with a small error.

\subsection{Actor Convergence Analysis} \label{subsec:actor_convergence}
\modified{2.5}{In this subsection, we study the convergence properties of the actor update in the NPG-AC algorithm, restricting the analysis to the class of linear state-feedback controllers. We note that for chance-constrained optimal control, the globally optimal policy is in general a nonlinear function of the state, even under linear dynamics, quadratic cost, and a linear risk function. However, the restriction to the linear-Gaussian policy class (\ref{eqn:policy_K}) enables a rigorous characterisation of the convergence behaviour of the actor updates in the chance-constrained setting. Accordingly, all optimality and convergence statements are to be understood with respect to the best policy within this class.} The following subsection describes the assumed policy structure. \modified{2.1.2}{Furthermore, we also assume ${\bf w}_k \sim \mathcal{N}(0,\Sigma_w)$, where $\Sigma_w >0$ is not known.} 
\subsubsection{Policy Structure} \label{subsubsec:policy_structure}
We consider a Gaussian policy with linear state-feedback as follows:
\begin{align}
    {\bf u} \sim \pi_{\bf K}  = \mathcal{N}(-{\bf Kx},\Sigma_{\sigma}).
    \label{eqn:policy_K}
\end{align}
Note that the policy parameter $\theta$ is replaced by the state-feedback gain matrix ${\bf K} \in \mathbb{R}^{p\times n}$. 

The set of stabilising state-feedback gain matrices is defined as
\begin{equation}
    \mathcal{K} \triangleq \left\{ {\bf K} \in \mathbb{R}^{p\times n} \mid \rho({\bf A} - {\bf B K}) < 1 \right\},
\end{equation}
where $\rho(\cdot)$ denotes the spectral radius of a matrix. Note that the set $\mathcal{K}$ is an open set. Furthermore, for control action given by (\ref{eqn:policy_K}), the closed-loop system dynamics is given by
\begin{equation}
    {\bf x}_{k+1} = ({\bf A} - {\bf B K}) {\bf x}_k + {\bar {\bf w}}_k, 
    \label{eqn:cl_sys}
\end{equation}
where $\bar {\bf w}_k$ is as follows,
\begin{equation}
    \bar {\bf w}_k = {\bf B} \sigma_k + {\bf w}_k, \quad \sigma_k \sim \mathcal{N}(0,\Sigma_{\sigma}).
\end{equation}
Therefore, $\bar {\bf w}_k \sim \mathcal{N}(0, \Sigma_{\bar w})$, where $\Sigma_{\bar w} = {\bf B} \Sigma_{\sigma} {\bf B}^T + \Sigma_w$, and ${\bf w}_k \sim \mathcal{N}(0,\Sigma_w)$.

Furthermore, for the convergence analysis we consider the following structure of the risk function $f_c(\cdot)$ in (\ref{eqn:J_c}).
\subsubsection{Risk Function Structure} \label{subsubsec:risk_function_structure}
We consider a linear function of the state as the risk function, \ie,
\begin{equation}
    f_c({\bf x}_{k+1}) = {\bf q}^T {\bf x}_{k+1},
    \label{eqn:fq}
\end{equation}
where ${\bf q} \in \mathbb{R}^n$ is a user defined vector. Note that this choice of risk function is standard in the chance-constrained control literature \cite{schildbachLinearControllerDesign2015}. \modified{3.3.2}{Since we are interested in bounding the tail probability of
${\bf q}^T{\bf x}_{k+1}$, which is zero-mean Gaussian in stationarity, we may
assume $\epsilon > 0$ without loss of generality. Indeed, the lower tail is
recovered by replacing ${\bf q}$ with $-{\bf q}$, as
$\mathbb{P}\{{\bf q}^T{\bf x}_{k+1} \le -\epsilon\} =
\mathbb{P}\{(-{\bf q})^T{\bf x}_{k+1} \ge \epsilon\}$, so both cases are
covered by the same formulation with $\epsilon > 0$.}

Next, we define a surrogate Lagrangian function in (\ref{eqn:JL}) and discuss the equivalance to the original one (\ref{eqn:JL}) in terms of the optimal point in the following subsection.

\subsubsection{Lagrangian Function Structure} \label{subsubsec:lagrangian_function_structure}
For the policy structure in (\ref{eqn:policy_K}) and the risk function structure in Subsection~{\ref{subsubsec:risk_function_structure}}, the LQR cost $J({\bf K})$ in (\ref{eqn:cost_fun_J}) takes the following closed-form expression.
\begin{align}
J(\bf K) &= \text{tr}\left( \left(\bf Q + {\bf K}^T{\bf R}{\bf K}  \right)\Sigma_K + {\bf R}\Sigma_{\sigma}  \right) \label{eqn:JK1} \\
    &= \text{tr}\left({\bf P}_K \Sigma_{\bar w} +  {\bf R}\Sigma_{\sigma} \right) \label{eqn:JK2} 
\end{align}
Note that if ${\bf K} \in \mathcal{K}$, then $\Sigma_K$ and ${\bf P}_K$ are the unique solutions to the Lyapunov equations given in (\ref{eqn:Sigma_K}) and (\ref{eqn:P_K}), respectively.
\begin{align}
    \Sigma_K &= \Sigma_{\bar w} + ({\bf A} - {\bf B}{\bf K})\Sigma_K({\bf A} - {\bf B}{\bf K})^T \label{eqn:Sigma_K} \text{, and} \\
    {\bf P}_K &= {\bf Q} + {\bf K}^T{\bf R}{\bf K} + ({\bf A} - {\bf B}{\bf K})^T{\bf P}_K({\bf A} - {\bf B}{\bf K}) \label{eqn:P_K}
\end{align}
\modified{3.3.3}{Furthermore, $\bar J_c({\bf K})$ can be expressed as follows.
\begin{align}
    \bar J_c({\bf K}) &= g\left(s({\bf K})\right), \text{ where } g(s) \triangleq Q_g\left(\epsilon/\sqrt{s}\right) \nonumber \\
    &\text{ and } s({\bf K}) \triangleq {\bf q}^T\Sigma_{\bf K}{\bf q}. \label{eqn:Jc_s1}
\end{align}
Here, $Q_g(a)$ ($\triangleq \mathbb{P}\left\{ z \ge a \right\} = \frac{1}{\sqrt{2\pi}}\int_{a}^{\infty} e^{-\frac{z^2}{2}}\,dz, \quad z \sim \mathcal{N}(0,1)$) is the standard normal tail function, which is strictly decreasing in its argument for $a > 0$. Therefore, $g(\cdot)$ is strictly increasing in $s$.

\begin{remark}
    The derivations of (\ref{eqn:JK1}), (\ref{eqn:JK2}), (\ref{eqn:Sigma_K}) and (\ref{eqn:P_K}) are available in \cite{yangProvablyGlobalConvergence2019}. It is straightforward to derive (\ref{eqn:Jc_s1}) using (\ref{eqn:cl_sys})-(\ref{eqn:fq}) in (\ref{eqn:J_c}) and using ${\bf x}_{k+1} \sim \mathcal{N}(0,\Sigma_{\bf K})$ in stationarity.
\end{remark}

Since $\bar J_c$ depends on ${\bf K}$ only through the scalar $s({\bf K})$, and $Q_g\big(\epsilon/\sqrt{s({\bf K})}\big)$ is strictly increasing in $s({\bf K})$, the chance constraint $\bar J_c({\bf K}) \le \delta$ is equivalent to
\begin{equation}
    s({\bf K}) = {\bf q}^T\Sigma_{\bf K}{\bf q} \le \epsilon_\delta, \qquad
    \epsilon_\delta \triangleq \frac{\epsilon^2}{\left(Q_g^{-1}(\delta)\right)^{2}},
    \label{eqn:equiv_constraint}
\end{equation}
which is linear in $\Sigma_{\bf K}$. Accordingly, for $\lambda \ge 0$ we define the surrogate Lagrangian with the same objective function and the modified but equivalent constraint as
\begin{equation}
    \mathcal{L}({\bf K},\lambda) \triangleq J({\bf K}) + \lambda\left(J_c({\bf K}) - \epsilon_\delta\right),
    \label{eqn:Lsur}
\end{equation}
By using (\ref{eqn:JK1}) and $J_c({\bf K}) = s({\bf K}) = \text{tr}({\bf q}{\bf q}^T\Sigma_{\bf K})$, it can be rewritten as
\begin{align}
    \mathcal{L}({\bf K},\lambda) &= \left\langle \widetilde{\bf Q} + {\bf K}^T{\bf R}{\bf K},\, \Sigma_{\bf K}\right\rangle + \left\langle {\bf R},\Sigma_\sigma\right\rangle - \lambda\epsilon_\delta, \nonumber \\
    &\text{ where } \widetilde{\bf Q} \triangleq {\bf Q} + \lambda\,{\bf q}{\bf q}^T \succ 0. \label{eqn:L_LQR}
\end{align}
Therefore, up to an additive constant, $\mathcal{L}(\cdot,\lambda)$ is itself an LQR cost with the modified state weight $\widetilde{\bf Q}$, which is what enables the convergence analysis in the following subsection.

\begin{remark} \label{rem:equivalence}
From this point onward, we will work with the surrogate Lagrangian (\ref{eqn:Lsur}) for the convergence analysis. Working with (\ref{eqn:Lsur}) does not alter the optimisation problem being solved since the constraints $\bar J_c({\bf K}) \le \delta$ and $J_c({\bf K}) \le \epsilon_\delta$ define the same feasible set. The constrained problem (\ref{eqn:opt_prob}) is therefore unchanged, and (\ref{eqn:JL}) and (\ref{eqn:Lsur}) attain their optimum at the same policy ${\bf K}^*$. Only the numerical value of the multiplier differs
\end{remark}
}
Next, we discuss the derived lemmas and the theorem regarding the convergence properties of the actor update in the NPG-AC algorithm.

\subsubsection{Actor Convergence Results} \label{subsubsec:convergence_results}
First, we establish the coercivity, L-smoothness, gradient dominance properties of the Lagrangian function $\mathcal{L}(\bf K,\lambda)$ in Lemma~\ref{lemma:coercivity}, Lemma~\ref{lemma:Lsmooth}, and Lemma~\ref{lemma:gradient_dominance}, respectively. Then, we prove the linear convergence of the actor update under the bounded error in the critic's value function approximation in Theorem~\ref{th:monotonicity}. Furthermore, we show that there is no duality gap in Lemma~\ref{lemma:duality}. 

The following lemma establishes the coercivity property of the Lagrangian function $\mathcal{L}(\bf K,\lambda)$ for a fixed $\lambda \ge 0$.
\begin{lemma}[Coercivity] \label{lemma:coercivity}
    For a fixed $\lambda > 0$, the Lagrangian function $\mathcal{L}(\bf K, \lambda)$ given by (\ref{eqn:Lsur}) is coercive on $\mathcal K$ in the sense that $\mathcal{L}(\bf K, \lambda) \rightarrow \infty$ as $\bf K \rightarrow \delta \mathcal{K}$, where $\delta \mathcal{K}$ denotes the boundary of $\mathcal K$. 
\end{lemma}
\begin{proof}
\modified{3.3.6}{By (\ref{eqn:L_LQR}), $\mathcal{L}(\cdot,\lambda)$ is, up to an additive constant, an LQR cost with state weight $\widetilde{\bf Q} = {\bf Q} + \lambda{\bf q}{\bf q}^T \succ 0$ and input weight ${\bf R} \succ 0$, and such a cost is coercive on $\mathcal{K}$ \cite{huTheoreticalFoundationPolicy2023}.}
\end{proof}
\begin{remark}
        The coercivity property of $\mathcal{L}(\bf K, \lambda)$ is crucial to ensure the stability of the closed-loop system during the training process. In other words, the coercive function $\mathcal{L}(\bf K, \lambda)$ serves as a barrier function over the stable policy set $\mathcal K$, and no additional measure is required to ensure the stability of the closed-loop system during the training process. Therefore, if we start from a stabilising policy, \ie, ${\bf K}_0$ and improve it at every iteration by keeping $\mathcal{L}({\bf K}_{j+1}, \lambda) \le \mathcal{L}({\bf K}_{j}, \lambda)$, then all the intermediate policies ${\bf K}_j, j = 0,1,\ldots$ will also be stabilising \cite{huTheoreticalFoundationPolicy2023}.
\end{remark}

For the convergence analysis of the actor update, L-smoothness and gradient dominance properties of the Lagrangian function $\mathcal{L}(\bf K, \lambda)$ are required. The following lemmas establish the L-smoothness and gradient dominance property of $\mathcal{L}(\bf K, \lambda)$ for a fixed $\lambda \ge 0$.

\begin{lemma}[Smoothness of the Lagrangian]\label{lemma:Lsmooth}
Fix $\lambda>0$ and define the sublevel set
\[
\mathcal K_\zeta \triangleq \left\{\, {\bf K}\in\mathcal K \;\middle|\; \mathcal L({\bf K},\lambda)\le \zeta \,\right\}.
\]
Then $\mathcal L(\cdot,\lambda)$ is $L$–smooth on $\mathcal K_\zeta$. Here $L >0$ is a constant that depends on the problem parameters, on $\zeta$, and on $\lambda$ through $\widetilde{\bf Q}$.
\end{lemma}
\begin{proof}
\modified{3.3.4}{By (\ref{eqn:L_LQR}), $\mathcal{L}(\cdot,\lambda)$ is, up to an additive constant, an LQR cost with state weight $\widetilde{\bf Q} = {\bf Q} + \lambda{\bf q}{\bf q}^T \succ 0$ and input weight ${\bf R} \succ 0$. Since additive constants do not affect smoothness, such a cost is twice continuously differentiable on the stabilising set $\mathcal{K}$ \cite{huTheoreticalFoundationPolicy2023}. Moreover, by Lemma~\ref{lemma:coercivity}, $\mathcal{L}(\cdot,\lambda)$ is coercive on $\mathcal{K}$, so the sublevel set $\mathcal{K}_\zeta$ is compact. Furthermore, $\nabla^2_{\bf K}\mathcal{L}({\bf K},\lambda)$ is continuous on $\mathcal{K}$, it is bounded on $\mathcal{K}_\zeta$, which gives the $L$-smoothness of $\mathcal{L}(\cdot,\lambda)$ on $\mathcal{K}_\zeta$ see also Theorem~1 of \cite{huTheoreticalFoundationPolicy2023}. The constant $L$ depends on the problem parameters, on $\zeta$, and on $\lambda$ through $\widetilde{\bf Q}$.}
\end{proof}


\begin{remark}
    The L-smoothness property of $\mathcal{L}(\bf K, \lambda)$ is crucial for the convergence analysis of the actor update in the NPG-AC algorithm. It ensures that the gradient of the Lagrangian function does not change too rapidly, which is essential for the stability and convergence of gradient-based optimisation methods. The constant $L$ provides a bound on how much the gradient can change, which helps in determining appropriate step sizes for the actor updates.
\end{remark}
Then we establish the gradient dominance property of the Lagrangian function $\mathcal{L}(\bf K, \lambda)$ in the following lemma.
\begin{lemma}[Gradient Dominance]\label{lemma:gradient_dominance}
\modified{3.3.7}{Fix $\lambda>0$. Then $\mathcal L(\cdot,\lambda)$ is gradient dominated on $\mathcal K$, \ie, for all ${\bf K}\in\mathcal K$,
\[\mathcal L({\bf K},\lambda) - \mathcal L({\bf K}^*,\lambda) \le \mu \|\nabla_{\bf K}\mathcal L({\bf K},\lambda)\|_F^2,\]
where ${\bf K}^* = \arg\min_{{\bf K}\in\mathcal K}\mathcal L({\bf K},\lambda)$ and the constant $\mu>0$ is given explicitly by
\[\mu^{-1} = \sigma_{\min}({\bf R})\,\sigma^3_{\min}(\Sigma_{\bar w})\,\|\Sigma_{{\bf K}^*}\|^{-2}.\]
In particular, the property holds on any sublevel set $\mathcal K_\zeta \triangleq \left\{\, {\bf K}\in\mathcal K \;\middle|\; \mathcal L({\bf K},\lambda)\le \zeta \,\right\}$ with the same constant.}
\end{lemma}
\begin{proof}
The proof of Lemma~\ref{lemma:gradient_dominance} is provided in Appendix~\ref{app:proof_lemma_gradient_dominance}.
\end{proof}
\begin{remark}
    The gradient dominance property ensures that the difference between the value of the Lagrangian function at any point and its minimum value can be bounded by the square of the norm of its gradient. This property is particularly useful in proving linear convergence rates for optimisation algorithms, as it provides a direct relationship between the function value and the gradient norm.
\end{remark}
Next, we show that, despite the presence of noise in the observed data and the use of finite samples to estimate the expected values required for the natural policy gradient in~\eqref{eqn:policy_update}, the Lagrangian function $\mathcal{L}(\mathbf{K}_j, \lambda)$ decreases monotonically with the iteration index $j$ with high probability, see Theorem~\ref{th:monotonicity}. To establish this, we need the following lemma regarding the convergence rate with true NPG. Here the term true NPG is used to refer to the natural policy gradient computed using the exact model parameters and the exact value function, without any approximation error. The true NPG value is given by ${\textbf{F}}^{-1}{G} = 2 ({\bf R} + {\bf B}^T{\bf P}_K{\bf B}) {\bf K} - 2{\bf B}^T{\bf P}_K{\bf A} $ \cite{yangProvablyGlobalConvergence2019}.

\begin{lemma}[Convergence rate with true NPG] For a given $\lambda > 0$, the NPG algorithm converge to a global optimal policy parameter ${\bf K}^*$ with a linear convergence rate, \ie,
    \begin{align}
        \mathcal{L}({\bf K}^{'}, \lambda) - \mathcal{L}({\bf K}^*, \lambda) \le \beta (\mathcal{L}({\bf K}, \lambda) - \mathcal{L}({\bf K}^*, \lambda)), \label{eqn:conv_rate}
    \end{align}
    provided we start from a stable controller, \ie, ${\bf K}_0 \in \mathcal{K}$. Here ${\bf K}^{'}$ is the NPG update from $\bf K$ in a single iteration (\ref{eqn:K_update}). The constant $0<\beta <1$ depends on the problem parameters and learning rate $0< \alpha < \frac{2}{L C_{\Sigma}}$, where $L$ is the smoothness constant of $\mathcal{L}(\cdot,\lambda)$ and $C_{\Sigma}$ is an upper bound on $tr(\Sigma_K)$.
    \begin{equation}
    {\bf K}^{'} = {\bf K} - \alpha {\textbf{F}}^{-1}{G}. \label{eqn:K_update}
\end{equation}
    Here ${{\bf F}}$ is the Fisher information matrix. $[{{\bf F}}]_{(i,j)(i^{'},j^{'})} = \text{E}[\nabla_{K_{ij}}\log(\pi_K({\bf u}|{\bf x}))\nabla_{K_{i^{'}j^{'}}}\log(\pi_K({\bf u}|{\bf x}))^T] $. 
    \label{lemma:conv_rate}
\end{lemma}
\begin{proof}
    The proof of Lemma~\ref{lemma:conv_rate} is provided in Appendix~\ref{app:proof_lemma_conv_rate}.
\end{proof}
\begin{remark}
    Lemma~\ref{lemma:conv_rate} implies that for any initial stabilising controller ${\bf K}_0\in\mathcal K$, the iterates $\{{\bf K}_i\}$ generated by \eqref{eqn:K_update} will converge to ${\bf K}^*$ as $i\to\infty$ at a linear rate determined by $\beta$.  
\end{remark}

Note that Lemma~\ref{lemma:conv_rate} holds for the exact natural policy gradient. However, in practice, we use the approximate natural policy gradient due to the error in the critic's value function approximation, noise in the observed data and the use of finite samples to estimate expected values. Therefore, we need to establish the monotonicity property of the Lagrangian function $\mathcal{L}({\bf K}_j, \lambda)$ under the approximate natural policy gradient. The following theorem establishes this property.
\begin{myth}[Convergence under approximate NPG] \label{th:monotonicity}
    For a fixed $\lambda > 0$, and under Corollary~\ref{cor:critic_error_bound}, for any initial stabilising controller ${\bf K}_0 \in \mathcal{K}$, $j \ge N_a$, where $N_a$ is sufficiently large, and step size $\alpha_j$ is sufficiently small that satisfies the condition in (\ref{eqn:cond_alpha_j}), then the iterates $\{{\bf K}_j\}$ generated by (\ref{eqn:policy_update}) satisfy
    \begin{align}
        &\mathcal{L}({\bf K}_{j+1}, \lambda) - \mathcal{L}({\bf K}^*, \lambda) \le \beta_1 (\mathcal{L}({\bf K}_j, \lambda) - \mathcal{L}({\bf K}^*, \lambda)), \nonumber \\
        & \text{ \wpp\ } 1 - \delta_K, \label{eqn:monotonicity}
    \end{align}
    where ${\bf K}_{j+1}$ is the policy update from ${\bf K}_j$ in a single iteration using (\ref{eqn:policy_update}), ${\bf K}^* = \arg\min_{{\bf K}\in\mathcal K}\mathcal L({\bf K},\lambda)$, $0<\beta_1<1, \delta_K <1$ depend on the problem parameters and learning rate $\alpha_j$. $\alpha_j$ is bounded as follows,
\begin{align}
&0 < \alpha_j < \min\left(\frac{2}{LC_\Sigma}, \alpha_j^*\right) \text{, where} \nonumber \\
&\alpha_j^* = \sup \{\alpha_j > 0: c_1 \alpha_j + c_2 \alpha_j^2 < (1-\beta)\left(\mathcal{L}({\bf K}_{j}, \lambda) \right. \nonumber \\
 &- \left. \mathcal{L}({\bf K}^*, \lambda) \right) \}.
 \label{eqn:cond_alpha_j}
\end{align}

Here, $c_1$ and $c_2$ depend on the problem parameters and state correlation matrix for the policy ${\bf K}_j$.
\end{myth}
\begin{proof}
    The proof of Theorem~\ref{th:monotonicity} is provided in Appendix~\ref{app:proof_th_monotonicity}.
\end{proof}
\begin{remark}[Learning rate condition]
    The upper bound on $\alpha_j$ in (\ref{eqn:cond_alpha_j}) has two components with distinct roles.
    The first component, $\frac{2}{LC_\Sigma}$, is inherited directly from the true NPG analysis in
    Lemma~\ref{lemma:conv_rate}. The condition ensures that a single true NPG step decreases the Lagrangian, and
    depends only on the L-smoothness constant $L$ and the upper bound $C_\Sigma$ on $\text{tr}(\Sigma_{{\bf K}_j})$.
    The second component, $\alpha_j^*$, guards against the additional error $\epsilon_L$ introduced by
    finite-sample estimation and critic approximation.
    Specifically, $\alpha_j^*$ is the largest step size for which the approximation-error term
    $\epsilon_L(\alpha_j) = c_1\alpha_j + c_2\alpha_j^2$ (see Appendix~\ref{app:proof_th_monotonicity})
    remains strictly smaller than $(1-\beta)$ times the current sub-optimality gap
    $\Delta_j \triangleq \mathcal{L}({\bf K}_j,\lambda) - \mathcal{L}({\bf K}^*,\lambda)$.
    For small $\alpha_j$, $\epsilon_L(\alpha_j) \approx c_1\alpha_j$, so $\alpha_j^*$ is approximately
    $\frac{(1-\beta)\Delta_j}{c_1}$.

    \emph{Practical satisfaction of the condition:}
    Because $\Delta_j$ is not available in closed form, the condition $\alpha_j < \alpha_j^*$ cannot
    be verified directly.  However, since $\epsilon_L(\alpha_j) \to 0$ as $\alpha_j \to 0$ while
    $\Delta_j > 0$ at any non-optimal iterate, a {sufficiently small constant step size} always
    satisfies (\ref{eqn:cond_alpha_j}).
    Concretely, any $\alpha$ with
    \[
        0 < \alpha \;<\; \min\!\left(\frac{2}{LC_\Sigma},\; \frac{(1-\beta)\Delta_{\min}}{c_1}\right),
    \]
    where $\Delta_{\min}$ is a lower bound on $\Delta_j$ over all iterates of interest, satisfies
    both constraints simultaneously.
    A lower bound $\Delta_{\min}$ can be obtained, for example, by running the algorithm for a warm-up
    phase and recording the smallest observed Lagrangian decrease.
    Alternatively, a \emph{decaying schedule} $\alpha_j = \alpha_0 / \sqrt{j+1}$ eventually satisfies
    the condition because $\alpha_j \to 0$ while $\Delta_j$ stays bounded away from zero until
    the algorithm is close to convergence.
    If the system matrices are known, the bounds on $L$ and $C_\Sigma$ that enter the first constraint can be estimated offline from the system matrices and a bound on the steady-state state covariance.
\end{remark}
\begin{remark}
    Theorem~\ref{th:monotonicity} establishes the linear convergence of the actor update in the NPG-AC algorithm under the bounded error in the critic's value function approximation. It states that after a sufficiently large number of iterations $N_a$, the Lagrangian function $\mathcal{L}({\bf K}_j, \lambda)$ will decrease monotonically with a rate determined by $\beta_1$ with high probability. The constant $\beta_1$ depends on the problem parameters and the learning rate $\alpha$. This result is crucial for ensuring that the actor update converges to a stable policy that satisfies the chance constraints.
\end{remark}
\subsubsection{Finding the optimal value of the Lagrange multiplier $\lambda$}
We follow a primal-dual approach to find an optimal value of the Lagrange multiplier $\lambda$, see Algorithm~\ref{algo:primal_dual}. The dual problem is defined as follows \vspace{-1.0mm}
\begin{align}
    \max_{\lambda \ge 0} D(\lambda) =  \max_{\lambda \ge 0} \min_{\bf K \in \mathcal K} \mathcal{L}(\bf K, \lambda) \label{eqn:dual_prob}.
\end{align}
\begin{algorithm}[t!]
    \caption{Primal-Dual Algorithm}
    \label{algo:primal_dual}
    \begin{algorithmic}[1]
        \State Initialise $\lambda_0$ and $\alpha_{\lambda,0}$.
        \For{$i=0,1,2,\ldots$}
            \State Solve the primal problem $\bf K_{i} = \operatorname*{arg\,min}_{\bf K \in \mathcal K} \mathcal{L}(\bf K, \lambda)$ using NPG, see Algorithm~\ref{alg:ac_method}.
            \State Evaluate $\nabla_{\lambda} \mathcal{L}(\bf K_i, \lambda) = J_c({\bf K}_i) - \epsilon_\delta$.
            \State Update $\lambda_{i+1} = \max\left(0, \lambda_i + \alpha_{\lambda,i} \nabla_{\lambda} \mathcal{L}(\bf K_i, \lambda) \right)$. 
        \EndFor
    \end{algorithmic}
\end{algorithm} 
In Algorithm~\ref{algo:primal_dual}, $\alpha_{\lambda,i} >0$, $\alpha_{\lambda,i} = \mathcal{O}(i^{-1/2})$, is the learning rate for the Lagrange multiplier $\lambda$. To prove that the pair ($\bf K^*, \lambda^*$) is also the optimal solution to the primal constrained problem (\ref{eqn:opt_prob}), we need Assumption~\ref{assump:slater} and Lemma~\ref{lemma:duality}.

\begin{assumption}[Slater's condition] There exists a $\bf \bar K \in \mathcal K$ such that $J_c(\bf \bar K) < \epsilon_\delta$.
    \label{assump:slater}
\end{assumption}
\begin{remark}
\modified{1.4.3}{Assumption~\ref{assump:slater} or the Slater's condition requires the existence of a strictly feasible stabilising controller, i.e., one whose constraint-violation probability is strictly below the threshold $\delta$, or equivalently, $J_c(\bf \bar K) < \epsilon_\delta$. This is a mild requirement as it only rules out the degenerate case in which the constraint can at best be met with equality. This assumption is important and standard for the primal--dual analysis because it provides a strict-feasibility condition required in establishing the strong-duality result used in Lemma~\ref{lemma:duality} \cite{zhaoGlobalConvergencePolicy2023}.}
\end{remark}
\begin{lemma}[Strong duality] Under
    Assumption~\ref{assump:slater}, the optimal value of the primal problem (\ref{eqn:opt_prob}) is equal to the optimal value of the dual problem (\ref{eqn:dual_prob}), \ie, $J^* = D^*$.
    \label{lemma:duality}
\end{lemma}
\begin{proof}
    The proof of Lemma~\ref{lemma:duality} is provided in Appendix~\ref{apdx:duality}.
\end{proof} {\color{black}
\begin{remark}[\textbf{Convergence of Algorithm~\ref{algo:primal_dual}}] Based on Lemma~\ref{lemma:coercivity} and Theorem~\ref{th:monotonicity}, we can say that the controller ${\bf K}_i$ from Algorithm~\ref{algo:primal_dual} will always be a stabilising controller. This implies that both $\nabla_{\lambda} \mathcal{L}(\bf K_i, \lambda)$ and $\lambda_i$ will be bounded by some positive constants. Furthermore, according to Theorem 4 in \cite{zhaoGlobalConvergencePolicy2023}, we can conclude that Algorithm~\ref{algo:primal_dual} will converge to an optimal policy at a sublinear rate, given that the step size $\alpha_{\lambda,i} = \mathcal{O}(i^{-1/2})$.
\end{remark}}

\section{NUMERICAL RESULTS} \label{sec:results}
\begin{figure*}[h!]
    \centering
    \subfloat[Policy gradient norm $||\hat G||$.\label{fig:policy_grad}]{\includegraphics[width=0.32\textwidth]{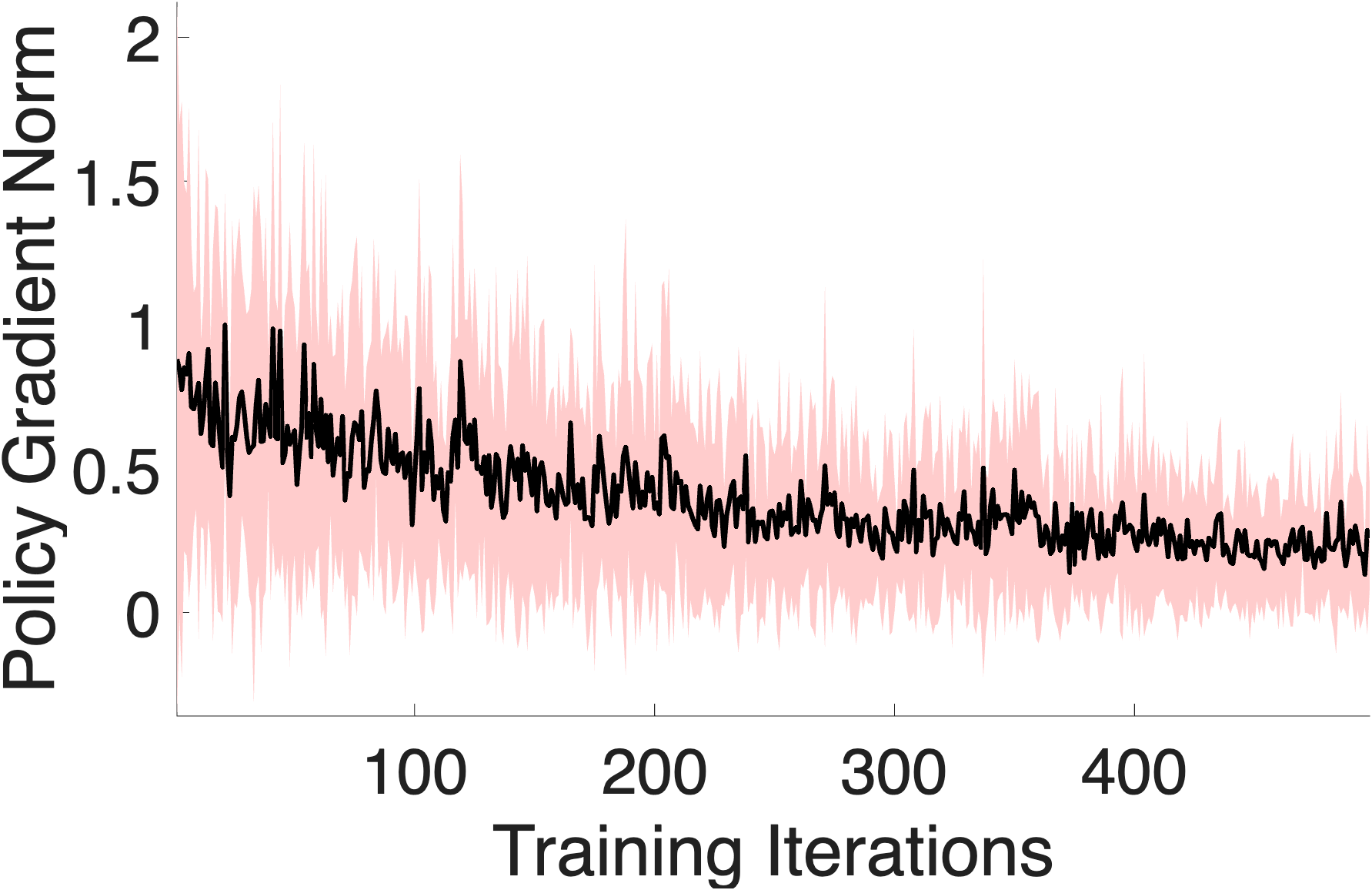}}\hfill%
    \subfloat[TD(0) critic loss.\label{fig:critic_loss}]{\includegraphics[width=0.32\textwidth]{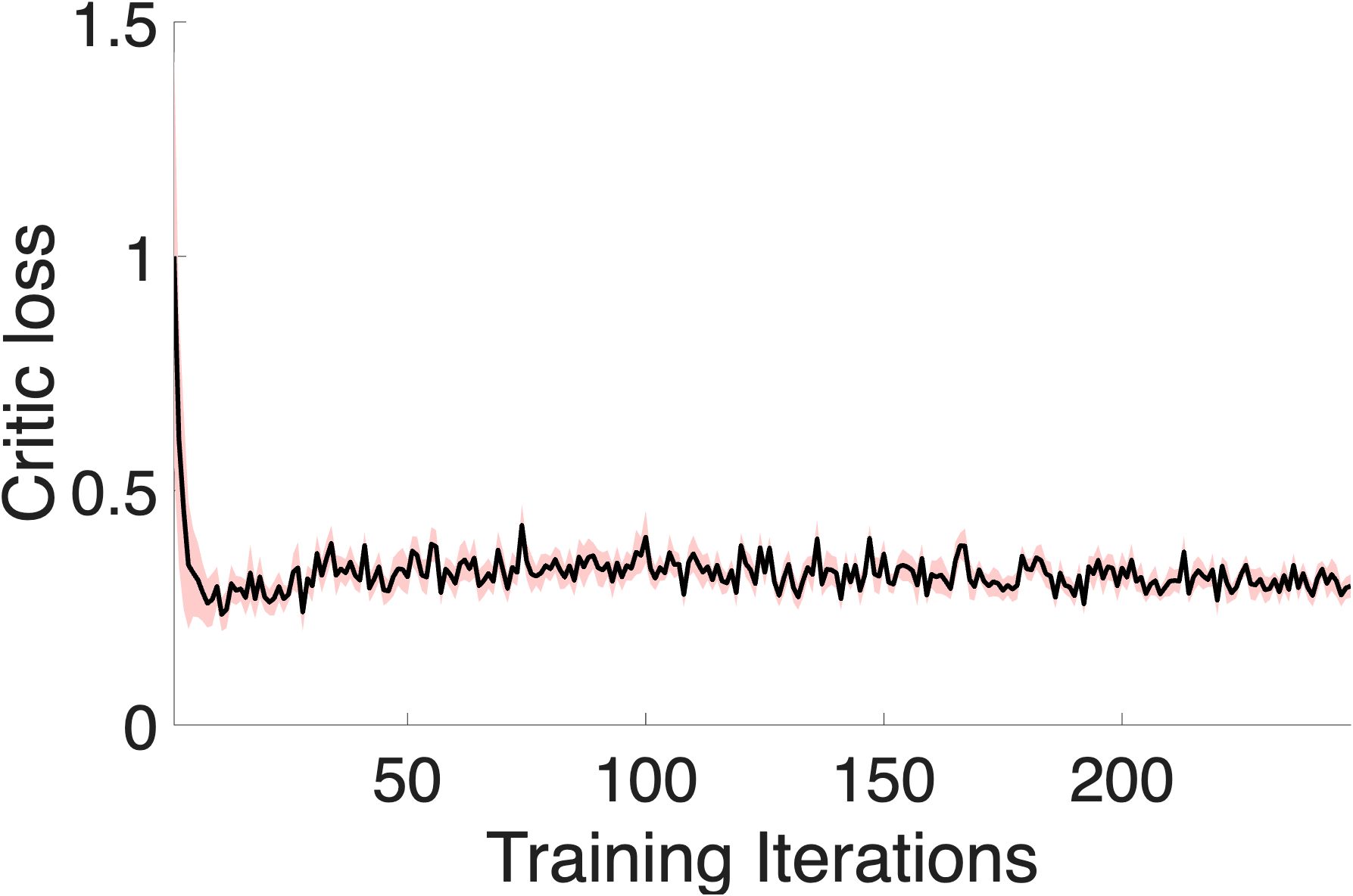}}\hfill%
    \subfloat[Average cost $\mathcal{L}$.\label{fig:avg_return}]{\includegraphics[width=0.32\textwidth]{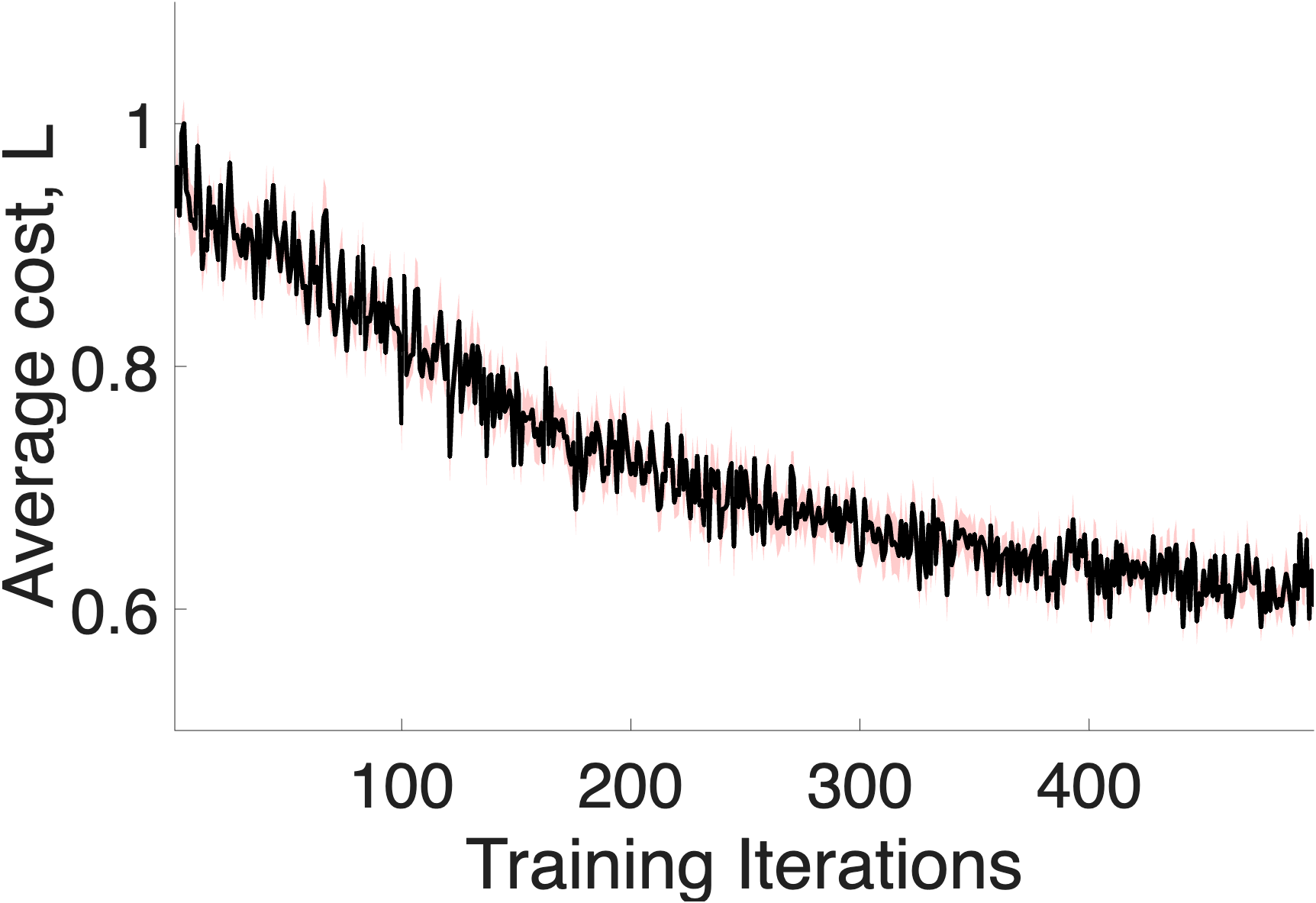}}%
    \caption{Training convergence vs.\ number of iterations (mean $\pm$ 90\% CI, 10 runs, normalised).}
    \label{fig:training_convergence}
\end{figure*}
In this section, we compare the performance of the proposed model-free NPG-based AC algorithm with model-based chance constrained LQR (CLQR), and scenario-based MPC through numerical simulations. \modified{1.6.1a}{As a case study, we consider an UAV model \cite{zhaoGlobalConvergencePolicy2023a}, which is a fourth-order LTI system ($n=4$, $p=2$). The system matrices $({\bf A},{\bf B})$ and the process-noise covariance $\Sigma_w$ are given in Appendix~\ref{apdx:params}. The control cost (\ref{eqn:cost_fun_J}) uses the weights ${\bf Q} = \mathrm{diag}(1,\,0.1,\,2,\,0.2)$ and ${\bf R} = {\bf I}_2$. The risk function is linear as in (\ref{eqn:fq}) with
${\bf q} = [1,\,0.1,\,2,\,0.2]^T$ and violation threshold $\epsilon = 5$. The per-stage
reward used for training is the same as in (\ref{eqn:rk})--(\ref{eqn:hrk}), and the negative of that reward is considered to be the cost. The remaining hyperparameters as mentioned in Algorithm~\ref{alg:ac_method} are listed
in Appendix~\ref{apdx:params}.} 

We use the quadratic feature space for the linear critic, \ie, $z(\mathbf{x}, \mathbf{u}) = \frac{1}{Z_c}
\begin{bmatrix}
\operatorname{svec}(\mathbf{x}\mathbf{x}^\top)^{\!\top} &
\operatorname{svec}(\mathbf{u}\mathbf{u}^\top)^{\!\top} &
\operatorname{vec}(\mathbf{u}\mathbf{x}^\top)^{\!\top} &
1
\end{bmatrix}^T$. Here, $\operatorname{svec}(\cdot)$ denotes the vectorisation of the upper-triangular entries of a symmetric matrix. Additionally, $Z_c$ is a normalisation constant which depends on the maximum possible values of $\mathbf{x}$ and $\mathbf{u}$ to ensure $\|z(\mathbf{x}, \mathbf{u})\|_2 \leq 1/2$. One can use other feature spaces, such as radial basis functions (RBF) or neural networks.  

\modified{1.6.3}{We compare against two model-based baselines. The first is the
chance-constrained LQR of \cite{schildbachLinearControllerDesign2015}
for the linear risk function (\ref{eqn:fq}), where the chance constraint is
reformulated as a quadratic constraint on the stationary state covariance,
and the design problem is transformed into a semidefinite program (SDP) whose
solution yields a static feedback gain. The complete formulation is given as
Algorithm~\ref{algo:SDP} in Appendix~\ref{apdx:algorithms}. We refer to this
baseline as CLQR (SDP-based), and it requires full knowledge of
$({\bf A},{\bf B})$ and of the noise distribution. The second baseline is the
scenario-based MPC of \cite{schildbachScenarioApproachStochastic2014}, which
at every time step draws $S$ noise scenarios and solves a finite-horizon
optimal control problem of horizon $T$ over the sampled scenarios
(Algorithm~\ref{algo:MPC} in Appendix~\ref{apdx:algorithms}). This scenario-based MPC algorithm also
requires full model knowledge. We use $S=20$ and $T=5$ for our simulation study. This operating point
is motivated by the computational budget since the per-step complexity of the
scenario program grows as $O(ST^2)$ \cite{skaf2009nonlinear}. Increasing
$S$ and $T$ improves the MPC performance at a rapidly increasing
computational cost.}


Fig.~\ref{fig:training_convergence} shows the policy gradient norm $||\hat G||$, TD(0) critic loss, and average cost $\mathcal{L}$ during training. Each plot shows the mean and 90\% confidence interval over 10 independent runs, normalised by the maximum value. \modified{1.6.2}{From Fig.~\ref{fig:training_convergence}, we can conclude that the critic loss decreases on average over the 10 trials, the policy-gradient norm diminishes, and the average cost consistently improves, which aligns with the theoretical results.}

In Fig.~\ref{fig:JvsJc}, we plot the control cost $J$ and the constraint violation probability $J_c$ for different values of the Lagrange multiplier $\lambda$. For CLQR, the $J_c$ values from NPG are used as thresholds, \ie, $\delta = J_c/100$. We observe that the proposed NPG-based AC method outperforms the MPC. However, it is crucial to note that MPC's performance is heavily dependent on the chosen parameters \( S = 20 \) and \( T = 5 \) in Algorithm~\ref{algo:MPC}. While increasing these parameters can enhance MPC's performance, it comes with the trade-off of increased computational complexity, which is of the order of $ST^2$ \cite{skaf2009nonlinear}. In addition, MPC necessitates solving an optimisation problem at every time step, in contrast to the proposed NPG-based AC method, which only requires evaluating the feed-forward actor and critic networks. This distinction renders the proposed method significantly less computationally complex than MPC. It is also important to acknowledge that MPC is a model-based method, which further differentiates it from the PG-based techniques. Additionally, the proposed NPG-based method performed very similarly to the CLQR method. Note that CLQR is also a model-based approach. 

\modified{1.6.4}{In Fig.~\ref{fig:comp_pri_dual}, we compare the primal-dual
NPG-based AC (Algorithm~\ref{algo:primal_dual}) with the CLQR (SDP-based)
baseline for the same threshold values $\delta$. Each reported point is an
empirical average computed over 10 independent trials. CLQR attains a slightly lower control cost, which is expected since it has
access to the exact model. The price paid by the proposed model-free method
is the training data. In return, it requires neither model knowledge nor
online optimisation, and it reaches within a small margin of the model-based
benchmark while meeting the same constraint threshold $\delta$.} 
\begin{figure}[h!]
    \centering
    \includegraphics[width=0.35\textwidth]{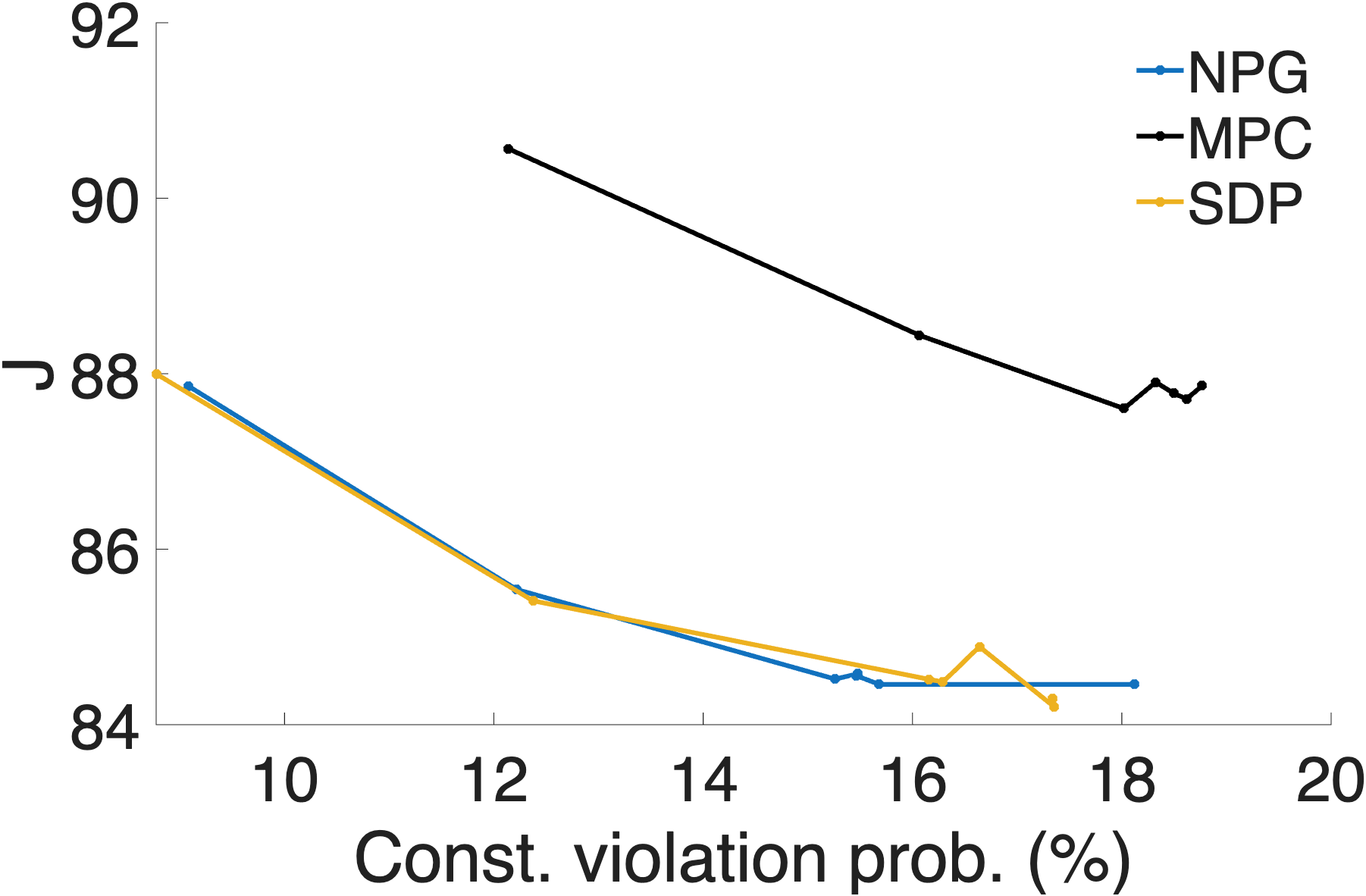} 
    \caption{Control cost $J$ vs. constraint violation probability $J_c$. $\lambda = [1,5,10,15,20,50,100]$}
    \label{fig:JvsJc}
\end{figure}
\begin{figure}[h!]
    \centering
    \includegraphics[width=0.35\textwidth]{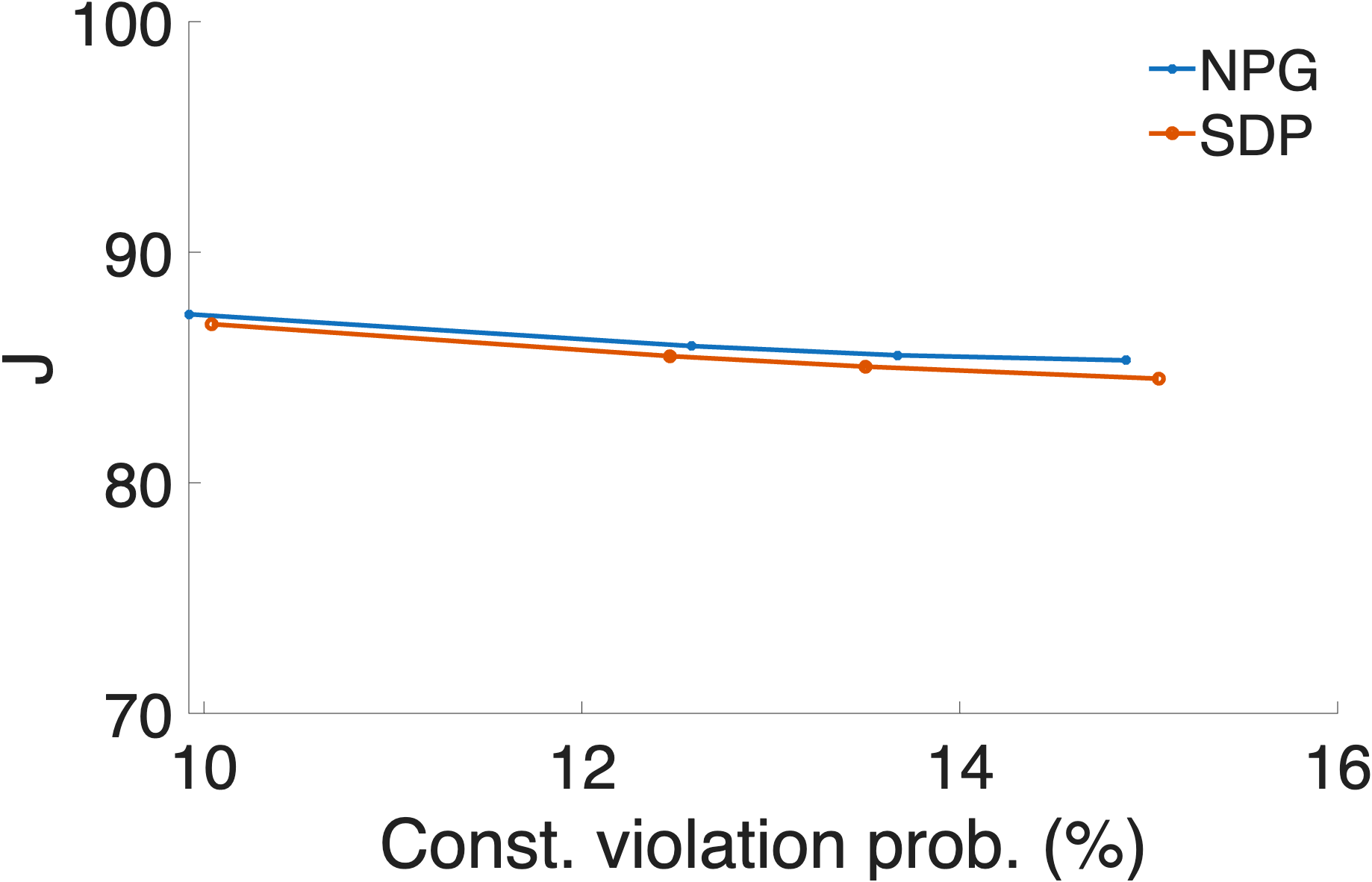} 
    \caption{Primal-dual NPG-based AC vs. CLQR for the same threshold $\delta$. $\delta = [15.0,13.5, 12.5, 10.0]$}
    \label{fig:comp_pri_dual}
\end{figure}

\section{CONCLUSION} \label{sec:conclusion}
This paper presents a reinforcement learning framework based on policy gradients for LQG control under chance constraints. We reformulate the constrained optimisation problem using Lagrangian relaxation and study a model-free solution through an actor–critic architecture. Within this model-free framework, we demonstrate that the approximate critic converges with high probability to a neighbourhood of the true critic, provided that appropriate regularity conditions are met and a suitable learning rate is chosen. Furthermore, to analyse the convergence of the actor, we study how errors propagate from the approximate critic and establish the coercivity, L-smoothness, and gradient dominance properties of the Lagrangian function under investigation. Based on these findings, we prove that the proposed AC method converges with high probability to a locally optimal stabilising policy after a sufficiently large number of iterations and with a small learning rate. Additionally, numerical results demonstrated that our method effectively satisfies the probabilistic constraints while achieving performance comparable to model-based CLQR and scenario-based MPC controllers, all without requiring knowledge of system matrices. \modified{2.2}{Even though the proposed algorithm does not require any linearity assumption on the critic and actor functions for its implementation, extending the convergence analysis and the empirical study to nonlinear function approximators, such as multi-layer NNs, is left for future work.} \vspace{-2mm}

\appendices
\section{Proof of Lemma~\ref{lemma:pd_Az}} \label{app:proof_lemma_pd_Az}
Since, $c_{\alpha}>0$, it is sufficient to show that the matrix $\bar{A}_z = \mathbb{E}[{\bf z}({\bf z}-{\bf z}')^T]$ is positive definite.
Let ${\bf v} \in \mathbb{R}^d$, ${\bf v} \neq {\bf 0}$. Consider the quadratic
form
\begin{equation}
{\bf v}^T \bar A_z {\bf v}
= \mathbb{E}\!\left[ ({\bf v}^T{\bf z}) ({\bf v}^T({\bf z}-{\bf z}')) \right].
\label{eq:quad_form_Az}
\end{equation}
Then, we define the scalar random variables $\psi$ and $\psi'$ as
$ \psi \triangleq {\bf v}^T{\bf z}$, and  $\psi' \triangleq {\bf v}^T{\bf z}'$, respectively.
Using these definitions, we can rewrite \eqref{eq:quad_form_Az} as
\begin{equation}
{\bf v}^T \bar A_z {\bf v}
= \mathbb{E}[\psi^2] - \mathbb{E}[\psi\psi'].
\label{eq:quad_form_Az_psi}
\end{equation}
Under the ergodicity assumption, we can write $\mathbb{E}[\psi^2] = \mathbb{E}[(\psi')^2]$, and
$\mathbb{E}[\psi\psi'] = \mathbb{E}[\psi'\psi]$. Then, we can say,
\begin{align}
\mathbb{E}\!\left[(\psi-\psi')^2\right]
= 2\big(\mathbb{E}[\psi^2] - \mathbb{E}[\psi\psi']\big).
\label{eq:exp_diff_psi}
\end{align}
Using, \eqref{eq:exp_diff_psi} in \eqref{eq:quad_form_Az_psi}, we get
\begin{equation}
\begin{aligned}
{\bf v}^T \bar A_z {\bf v}
= \frac{1}{2}\,\mathbb{E}\!\left[(\psi-\psi')^2\right]
= \frac{1}{2}\,\mathbb{E}\!\left[\big({\bf v}^T({\bf z}-{\bf z}')\big)^2\right]
\;\ge\; 0.
\end{aligned}
\label{eq:quad_form_Az_final}
\end{equation}
Moreover, ${\bf v}^T \bar A_z {\bf v} = 0$ if and only if
${\bf v}^T({\bf z}-{\bf z}') = 0$. Since the process noise is assumed to be
Gaussian and the policy is stochastic, ${\bf z} \neq {\bf z}'$ almost surely. Consequently, ${\bf v}^T \bar A_z {\bf v} > 0 \quad \forall\,{\bf v}\neq {\bf 0}$, almost surely. Therefore, $\bar A_z$ and $\bar A$ are positive definite almost surely. This completes the proof. \vspace{-2mm}

\section{Proof of Theorem~\ref{thm:critic_convergence}} \label{app:proof_thm_critic_convergence}
    The proof follows similar steps as in \cite{dalal2018finite}. Here, we only present an updated lemma and its proofs related to the martingale noise $M_{l+1}$, which is different from those in \cite{dalal2018finite}. The rest of the proof follows the same steps as in \cite{dalal2018finite} and is omitted here for brevity.
    \begin{lemma} \label{lemma:martingale_bound}
    For all $l \ge 0$, the martingale difference noise $M_{l+1}$ defined in (\ref{eqn:martingale_noise}) satisfies
    \begin{align}
        \mid \mid M_{l+1} \mid \mid &\le K_1 (1 + \mid \mid \Phi_l - \Phi^* \mid \mid) \text{, where} \nonumber \\
        K_1 &= \max\left\{c_{\alpha} + \frac{1}{2} + \sqrt{c^2_\alpha + \frac{5}{16}} \mid\mid \bar{A}^{-1} \mid\mid \mid \mid \bar{b} \mid\mid, \right. \nonumber \\
        &\quad \left. \sqrt{c^2_\alpha + \frac{5}{16}} + \mid\mid \bar{A} \mid\mid  \right\}.\label{eqn:martingale_bound}
    \end{align}
    \end{lemma}
    \begin{proof}
    Using the definition of $M_{l+1}$ in (\ref{eqn:martingale_noise}) and using $\bar b = \bar{A}\Phi^*$, we can write
    \begin{align}
        M_{l+1} &= A({\bf X}) \Phi_l + b({\bf X}) - \left( \bar{b} - \bar{A} \Phi_l \right) \label{eqn:apdx_ml} \\
        &= A({\bf X})\Phi^* + b({\bf X}) + \left(A({\bf X}) + \bar{A} \right) \left(( \Phi_l - \Phi^* \right)\nonumber
    \end{align}
    Now, taking $\mid\mid \cdot \mid\mid$ on both sides of (\ref{eqn:apdx_ml}), we can write
    \begin{equation}
    \begin{aligned}
    \mid \mid M_{l+1}  \mid \mid \le & \mid \mid A({\bf X})\mid \mid\mid \mid\Phi^*\mid \mid + \mid \mid b({\bf X}) \mid \mid \\
    & + \mid \mid A({\bf X}) + \bar{A} \mid \mid \mid \mid \Phi_l - \Phi^* \mid \mid 
    \end{aligned}
    \end{equation}
    Using (\ref{eqn:A_b_matrices}) and Assumption \ref{assumption:bounded_noise}, we can write
    \begin{align}
    \mid \mid b({\bf X}) \mid \mid  & = \sqrt{ c^2_\alpha r^2 + r^2 \mid \mid {\bf z} \mid \mid^2} \le c_{\alpha} + \frac{1}{2} \text{, and} \nonumber \\
    \mid \mid A({\bf X}) \mid \mid & \le  \mid \mid \begin{bmatrix} c_{\alpha} & 0 \\ \mid \mid {\bf z}\mid \mid & \mid \mid{\bf z}({\bf z}' - {\bf z})^T\mid \mid \end{bmatrix} \mid \mid \le \sqrt{c^2_\alpha + \frac{5}{16}}
    \label{eqn:bound_A_b}
    \end{align}
    Using the bounds from (\ref{eqn:bound_A_b}) in (\ref{eqn:apdx_ml}), Lemma~\ref{lemma:martingale_bound} follows.
\end{proof}

The proof of Theorem~\ref{thm:critic_convergence} in \cite{dalal2018finite} does not depend on the specific structure of the matrices $A\left({\bf X}\right)$, $b({\bf X})$, and the martingale noise $M_{l+1}$, but the norm of those quantities must be bounded as in Lemma~\ref{lemma:martingale_bound}. Therefore, by replacing Lemma~5.1 from \cite{dalal2018finite} with Lemma~\ref{lemma:martingale_bound} above, the rest of the proof of Theorem~\ref{thm:critic_convergence} follows the same steps as in \cite{dalal2018finite}.  \vspace{-2mm}

\section{Proof of Lemma~\ref{lemma:gradient_dominance} [Gradient Dominance]} \label{app:proof_lemma_gradient_dominance}
\modified{3.3.5}{
By (\ref{eqn:L_LQR}), up to an additive constant, $\mathcal{L}(\cdot,\lambda)$ is an LQR cost with state weight $\widetilde{\bf Q} = {\bf Q} + \lambda{\bf q}{\bf q}^T$ and input weight ${\bf R}$. The additive constant affects neither the sub-optimality gap nor the gradient. Since ${\bf Q} \succ 0$ and $\lambda{\bf q}{\bf q}^T \succeq 0$, we have $\widetilde{\bf Q} \succ 0$, so $({\bf A},\widetilde{\bf Q}^{1/2})$ is detectable, $({\bf A},{\bf B})$ is stabilisable by assumption, and $\Sigma_{\bar w} = {\bf B}\Sigma_\sigma{\bf B}^T + \Sigma_w \succ 0$. Corollary~2 of \cite{watanabe2026gradient} therefore applies to $\mathcal{L}(\cdot,\lambda)$ and yields, for every ${\bf K} \in \mathcal{K}$,
\begin{align}
    \mathcal{L}({\bf K},\lambda) - \mathcal{L}({\bf K}^*,\lambda) &\le \mu \mid\mid \nabla_K \mathcal{L}({\bf K},\lambda)\mid\mid^2, \nonumber \\
    \mu(\lambda)^{-1} = \sigma_{\min}({\bf R})\,&\sigma^3_{\min}(\Sigma_{\bar w})\mid\mid\Sigma_{{\bf K}^*}\mid\mid^{-2}, \label{eqn:gd_L}
\end{align}
where ${\bf K}^* = \arg\min_{{\bf K}\in\mathcal{K}}\mathcal{L}({\bf K},\lambda)$. \newline 
For completeness, we briefly outline how (\ref{eqn:gd_L}) is obtained in \cite{watanabe2026gradient}. Although the LQR cost is nonconvex in ${\bf K}$, a reformulation with an extended convex lifting procedure renders it convex, and this hidden convexity is used to establish the following two inequalities that bound the sub-optimality gap from below and from above,
\begin{align}
    \mathcal{L}({\bf K},\lambda) - \mathcal{L}({\bf K}^*,\lambda) &\ge \left\langle {\bf R}, ({\bf K}-{\bf K}^*)\Sigma_{\bf K}({\bf K}-{\bf K}^*)^T \right\rangle, \label{eqn:qg} \\
    \mathcal{L}({\bf K},\lambda) - \mathcal{L}({\bf K}^*,\lambda) &\le \left\langle \nabla_K \mathcal{L}({\bf K},\lambda)\,\Sigma_{\bf K}^{-1}\Sigma_{{\bf K}^*},\, {\bf K}-{\bf K}^* \right\rangle. \label{eqn:ub}
\end{align}
The first is a quadratic-growth property, and the second follows from the convexity of the reformulated problem, both are stated in the original parameterisation ${\bf K}$. From (\ref{eqn:qg}), $\mathcal{L}-\mathcal{L}^* \ge \sigma_{\min}({\bf R})\sigma_{\min}(\Sigma_{\bf K})\mid\mid{\bf K}-{\bf K}^*\mid\mid^2$, while (\ref{eqn:ub}) with the Cauchy--Schwarz inequality gives $\mathcal{L}-\mathcal{L}^* \le \mid\mid\Sigma_{{\bf K}^*}\mid\mid\sigma^{-1}_{\min}(\Sigma_{\bf K})\mid\mid{\bf K}-{\bf K}^*\mid\mid\,\mid\mid\nabla_K\mathcal{L}\mid\mid$. Eliminating $\mid\mid{\bf K}-{\bf K}^*\mid\mid$ between the two bounds yields (\ref{eqn:gd_L}) with $\sigma_{\min}(\Sigma_{\bf K})$ in place of $\sigma_{\min}(\Sigma_{\bar w})$, \ie, with a policy-dependent constant. Finally, in the discrete-time case, $\Sigma_{\bf K} = \sum_{t\ge0}({\bf A}-{\bf B}{\bf K})^t\Sigma_{\bar w}\left(({\bf A}-{\bf B}{\bf K})^T\right)^t \succeq \Sigma_{\bar w}$ for every ${\bf K}\in\mathcal{K}$, so that $\sigma_{\min}(\Sigma_{\bf K}) \ge \sigma_{\min}(\Sigma_{\bar w})$. This replaces the policy-dependent constant by the uniform one in (\ref{eqn:gd_L}) and makes the bound global on $\mathcal{K}$.} \vspace{-2mm}

\vspace{-1mm} \section{Proof of Lemma~\ref{lemma:conv_rate} [Convergence Rate]} \label{app:proof_lemma_conv_rate}
Here, we will prove Lemma~\ref{lemma:conv_rate} for the NPG algorithm. The update rule for the policy parameter ${\bf K}$ under the NPG algorithm is given by (\ref{eqn:K_update}).

From the L-smoothness property of $\mathcal{L}({\bf K},\lambda)$ as given in Lemma~\ref{lemma:Lsmooth}, we can write the following inequality \cite{huTheoreticalFoundationPolicy2023},
\begin{align}
    &\mathcal{L}({\bf K}^{'}, \lambda) - \mathcal{L}({\bf K}, \lambda) \le \nonumber \\
    &\text{tr}(\nabla_K\mathcal{L}({\bf K}, \lambda)^T({\bf K}^{'} - {\bf K})) + \frac{L}{2} \mid \mid {\bf K}^{'} - {\bf K} \mid \mid_F^2, \label{eqn:grad_dom2}
\end{align}
Using (\ref{eqn:K_update}) in (\ref{eqn:grad_dom2}), we can write
\begin{equation}
\begin{aligned}
    &\mathcal{L}({\bf K}^{'}, \lambda) - \mathcal{L}({\bf K}, \lambda) \le  -\text{tr}\left(\alpha {\Sigma_K}^{-1} \right.\\
    & \left. - \frac{L\alpha^2}{2}({\Sigma_K}^{-1})({\Sigma_K}^{-1})^T \right)\mid \mid \nabla_K \mathcal{L}({\bf K}, \lambda)  \mid \mid_F^2. \label{eqn:grad_dom3}
\end{aligned}
\end{equation}
We have used the matrix trace inequality as given in Theorem~1 from \cite{coope1994matrix} to get (\ref{eqn:grad_dom3}). To ensure convergence, we need the trace in (\ref{eqn:grad_dom3}) to be strictly positive. In other words, the step size $\alpha$ should be $\alpha < \frac{2}{L\text{tr}(\Sigma_K^{-1})}$.
Since $\Sigma_K$ is the solution to the Lyapunov equation (\ref{eqn:Sigma_K}), we can say $\text{tr}(\Sigma_K)$ is upper bounded by a finite constant, so there exists a constant $0<C_{\Sigma} < \text{tr}(\Sigma_K^{-1})$. Therefore, we can write an upper limit for $\alpha < \frac{2}{LC_{\Sigma}}$, which is independent of $\bf K$.

Applying the gradient dominance property of $\mathcal{L}({\bf K}, \lambda)$ as given in Lemma~\ref{lemma:gradient_dominance}, we can write 
\begin{align}
    &\mathcal{L}({\bf K}^{'}, \lambda) - \mathcal{L}({\bf K}^*, \lambda) \le \beta (\mathcal{L}({\bf K}, \lambda) - \mathcal{L}({\bf K}^*, \lambda)) \text{, where} \nonumber \\
    &\beta =1- \frac{1}{\mu}\text{tr}\left(\alpha {\Sigma_K}^{-1} - \frac{L\alpha^2}{2}({\Sigma_K}^{-1})({\Sigma_K}^{-1})^T \right)
\end{align}
Since we need $0<\beta < 1$ for convergence, the step size $\alpha$ should satisfy the following condition
\begin{align}
    &0 < \frac{1}{\mu}\text{tr}\left(\alpha {\Sigma_K}^{-1} - \frac{L\alpha^2}{2}({\Sigma_K}^{-1})({\Sigma_K}^{-1})^T \right) < 1 \nonumber \\
    &=> \frac{1}{\mu}\text{tr}\left(\alpha {\Sigma_K}^{-1} - \frac{L\alpha^2}{2}({\Sigma_K}^{-1})({\Sigma_K}^{-1})^T \right) < 1 \\ 
    &=> \frac{1}{\mu}\text{tr}\left(\frac{\alpha}{\sigma_{\min}(\Sigma_{\bar w})} - \frac{L\alpha^2}{2}C_{\Sigma}^2 \right) < 1 \text{, [(\ref{eqn:Sigma_K}) used]}
    \label{eqn:alpha_bound2}
\end{align}
In (\ref{eqn:alpha_bound2}), $\alpha < \frac{2}{LC_{\Sigma}}$ is assumed. Here, $\sigma_{\min}(\cdot)$ denotes the lowest singular value. Note that if we choose $\alpha$ sufficiently small, condition (\ref{eqn:alpha_bound2}) can be satisfied. This completes the proof of Lemma~\ref{lemma:conv_rate}. \vspace{-2mm}
\vspace{-1mm}\section{PROOF OF THEOREM~\ref{th:monotonicity} [Convergence of Policy Gradient]} \label{app:proof_th_monotonicity}
To prove Theorem~\ref{th:monotonicity}, we need the following lemma,
\begin{lemma}\label{lemma:K_diff}
For a fixed $\lambda > 0$, let the step size $\alpha$ in (\ref{eqn:policy_update}) be chosen such that Lemma~\ref{lemma:conv_rate} holds. Furthermore, let the error in the critic's value function approximation be bounded, \ie, Corollary~\ref{cor:critic_error_bound} holds. Then, for a sufficiently large number of samples $N$ used in the policy gradient estimation and a sufficiently large number of samples $M$ used in the state covariance matrix estimation, and for $\epsilon_K > 0$ and $\delta_K \in (0,1)$, the difference between the policy parameters updated in one time step using the true natural policy gradient, ${\bf K}'_{j+1}$, and updated using the approximate natural policy gradient, ${\bf K}_{j+1}$, both starting from the same parameter ${\bf K}_j$ in the previous step, is bounded as
\begin{equation}
    ||{\bf K}_{j+1} - {\bf K}^{'}_{j+1}|| \le \epsilon_K, \text{\wpp\ } 1 - \delta_K.
\end{equation}
Here $\epsilon_K$ is a constant that depends on the problem parameters.
\end{lemma}
\begin{remark}
    Lemma~\ref{lemma:K_diff} implies that estimation and approximation errors do not significantly distort the policy update direction or magnitude in each iteration, provided that sufficient samples are used. This is crucial for ensuring that the policy updates remain effective and converge towards the optimal policy, despite the inherent uncertainties in the estimation process. This lemma is instrumental in establishing the overall convergence of the policy gradient method in Theorem~\ref{th:monotonicity}.
\end{remark}
\begin{proof}
To differentiate between a policy update step with true natural policy gradient and an approximate natural policy gradient, we denote the policy parameter updated with true natural policy gradient as ${\bf K}^{'}$ and the one updated with approximate natural policy gradient as ${\bf K}$. The update rules for ${\bf K}^{'}$ and ${\bf K}$ are given by from (\ref{eqn:K_update}) as
\begin{align}
    &{\bf K}^{'}_{j+1} = {\bf K}_j - \alpha_j  \nabla_K \mathcal{L}({\bf K}_j, \lambda) \Sigma_{{\bf K}_j}^{-1}, \label{eqn:K_update_true} \\
    &{\bf K}_{j+1} =  {\bf K}_j - \alpha_j  \hat{\nabla}_K \mathcal{L}({\bf K}_j, \lambda) \hat{\Sigma}_{{\bf K}_j}^{-1}. \label{eqn:K_update_approx}
\end{align}
Subtracting (\ref{eqn:K_update_true}) from (\ref{eqn:K_update_approx}), and adding and subtracting a few same terms, we can write
\begin{align}
    &{\bf K}_{j+1} - {\bf K}^{'}_{j+1} =-\alpha_j \left( \hat{\nabla}_K \mathcal{L}({\bf K}_j, \lambda) (\hat{\Sigma}_{{\bf K}_j}^{-1} - {\Sigma}_{{\bf K}_j}^{-1}) \right. \nonumber \\
    & \left. + (\hat \nabla_K \mathcal{L}({\bf K}_j, \lambda) - \nabla_K \mathcal{L}({\bf K}_j, \lambda)) {\Sigma}_{{\bf K}_j}^{-1} \right).
 \label{eqn:K_diff}
\end{align}

We estimate $\hat{\Sigma}_{{\bf K}_j}$ as the following sample mean,
\begin{align}
    \hat{\Sigma}_{{\bf K}_j} = \frac{1}{M} \sum_{m=1}^{M} {\bf x}_m {\bf x}_m^T. \label{eqn:Sigma_hat}
\end{align}
Since we assume that we draw iid samples from the buffer, we can write the following bound for $\hat{\Sigma}_{{\bf K}_j}$, using the matrix concentration inequality \cite{tropp2015introduction}, 
\begin{align}
    &\mathbb{P}\left( ||\hat{\Sigma}_{{\bf K}_j} - \Sigma_{{\bf K}_j}|| \ge \epsilon \right) \le 2n \exp\left(\frac{-M\epsilon^2}{2C_x(\mid \mid \Sigma_{{\bf K}_j} \mid \mid + 2\epsilon/3)} \right). \label{eqn:matrix_conc_ineq}
\end{align}
Here, we assume $\mid \mid {\bf x}_m \mid \mid \le C_x$ almost everywhere for some $C_x > 0$. Since the CL-system will always remain stable, such an assumption is reasonable. From (\ref{eqn:matrix_conc_ineq}), we can say that for a sufficiently large $M$, $||\hat{\Sigma}_{{\bf K}_j} - \Sigma_{{\bf K}_j}||$ can be made arbitrarily small with high probability. Furthermore, using the perturbation theory of matrix inverse \cite{stewart1977perturbation}, we can write
\begin{align}
    &||\hat{\Sigma}_{{\bf K}_j}^{-1} - {\Sigma}_{{\bf K}_j}^{-1}|| \le ||{\Sigma}_{{\bf K}_j}^{-1}|| \frac{||\hat{\Sigma}_{{\bf K}_j} - {\Sigma}_{{\bf K}_j}||}{1 - ||{\Sigma}_{{\bf K}_j}^{-1}|| ||\hat{\Sigma}_{{\bf K}_j} - {\Sigma}_{{\bf K}_j}||}. \label{eqn:inv_perturb}
\end{align}
From (\ref{eqn:matrix_conc_ineq}) and (\ref{eqn:inv_perturb}), we can say that for a sufficiently large $M$, $||\hat{\Sigma}_{{\bf K}_j}^{-1} - {\Sigma}_{{\bf K}_j}^{-1}||$ can be made arbitrarily small with high probability.

Moreover, $\hat{\nabla}_K \mathcal{L}({\bf K}_j, \lambda) =\frac{1}{N} \sum_{k=1}^N {\hat Q}\left( {\bf x}_k, {\bf u}_k \right) \nabla_K \log \pi_K\left({\bf u}_k \mid {\bf x}_k \right) $ is bounded with high probability for the following reasons. First, from Corollary~\ref{cor:critic_error_bound}, we can say that $\hat Q\left( {\bf x}_k, {\bf u}_k \right)$ will be in close proximity of true $Q$ with probability $1-\delta_c$, and true $Q$ for a stable system will be finite \cite{yangProvablyGlobalConvergence2019}. Second, $\frac{1}{N} \sum_{k=1}^N  \nabla_K \log \pi_K\left({\bf u}_k \mid {\bf x}_k \right)$ is also bounded with high probability, \ie, $1-\delta_1$, see (\ref{eqn:grad_log_pi2}).

Therefore, we can write the following inequality, 
\begin{align}
    &||{\bf K}_{j+1} - {\bf K}^{'}_{j+1}|| \le \alpha_j ||\hat \nabla_K \mathcal{L}({\bf K}_j, \lambda) - \nabla_K \mathcal{L}({\bf K}_j, \lambda)|| \nonumber \\
     &
    \times ||{\Sigma}_{{\bf K}_j}^{-1}|| + \epsilon_\Sigma. \label{eqn:K_diff2}
\end{align}
Here, $\epsilon_\Sigma$ is a constant that depends on the problem parameters and can be made arbitrarily small by choosing sufficiently large $M$ and step size $\alpha_j$ sufficiently small. Therefore, we focus on deriving an upper bound for $||\hat \nabla_K \mathcal{L}({\bf K}_j, \lambda) - \nabla_K \mathcal{L}({\bf K}_j, \lambda)||$ in the following.

Next we focus on the difference between $\hat \nabla_K \mathcal{L}({\bf K}_j, \lambda)$ and $\nabla_K \mathcal{L}({\bf K}_j, \lambda)$. From (\ref{eqn:policy_gradient}), and adding and subtracting a few same terms, we can write the following expression, see (\ref{eqn:grad_diff}). Note that $\hat Q$ is the approximate Q-function obtained from $N_c$ iterations of the critic update, but the reference to $N_c$ is omitted for notational simplicity. Additionally, we have assumed the process to be ergodic under the policy $\pi_{{\bf K}_j}$, so the states will have a stationary distribution $\rho_x$.
\begin{align}
    &\hat \nabla_K \mathcal{L}({\bf K}_j, \lambda) - \nabla_K \mathcal{L}({\bf K}_j, \lambda) = \nonumber \\
    & = \frac{1}{N} \sum_{k=1}^N \left( {\hat Q}\left( {\bf x}_k, {\bf u}_k \right) - {Q}\left( {\bf x}_k, {\bf u}_k \right) \right) \nabla_K \log \pi_K\left({\bf u}_k \mid {\bf x}_k \right) \nonumber \\
    & + \frac{1}{N} \sum_{k=1}^N {Q}\left( {\bf x}_k, {\bf u}_k \right) \nabla_K \log \pi_K\left({\bf u}_k \mid {\bf x}_k \right) - \nonumber \\
    & \mathbb{E}_{x \sim \rho_x,u \sim \pi_{K_j}} \left[ {Q}\left( {\bf x}, {\bf u} \right) \nabla_K \log \pi_K\left({\bf u} \mid {\bf x} \right) \right].
     \label{eqn:grad_diff}
\end{align}
From (\ref{eqn:grad_diff}), we can write the following inequality,
\begin{align}
    &||\hat \nabla_K \mathcal{L}({\bf K}_j, \lambda) - \nabla_K \mathcal{L}({\bf K}_j, \lambda)|| \le \nonumber \\
    & \underbrace{\frac{1}{N} \sum_{k=1}^N \mid  {\hat Q}\left( {\bf x}_k, {\bf u}_k \right) - {Q}\left( {\bf x}_k, {\bf u}_k \right) \mid \mid \mid \nabla_K \log \pi_K\left({\bf u}_k \mid {\bf x}_k \right) \mid \mid}_{T_1} \nonumber \\
    & + \left| \left| \frac{1}{N} \sum_{k=1}^N {Q}\left( {\bf x}_k, {\bf u}_k \right) \nabla_K \log \pi_K\left({\bf u}_k \mid {\bf x}_k \right) \right. \right. \nonumber \\
    &\underbrace{\indent \indent \left. \left. - \mathbb{E}_{x \sim \rho_x,u \sim \pi_{K_j}} \left[ {Q}\left( {\bf x}, {\bf u} \right) \nabla_K \log \pi_K\left({\bf u} \mid {\bf x} \right) \right] \right| \right|}_{T_2}.
    \label{eqn:grad_diff2}
\end{align}
First we derive the upper bound for the first term $T_1$ on the right-hand side of (\ref{eqn:grad_diff2}). From Corollary~\ref{cor:critic_error_bound}, we can write
\begin{align}
    T_1\le &\left(\frac{\epsilon}{2} + M_Q \right) \frac{1}{N} \sum_{k=1}^N \mid \mid \nabla_K \log \pi_K\left({\bf u}_k \mid {\bf x}_k \right) \mid \mid, \nonumber \\
    & \text{\wpp\  } 1 - \delta_c \label{eqn:grad_diff_term1}
\end{align}
Using (\ref{eqn:policy_K}), we can simplify $\nabla_K \log \pi_K\left({\bf u}_k \mid {\bf x}_k \right)$ as
\begin{align}
    &\frac{1}{N} \sum_{k=1}^N \mid \mid \nabla_K \log \pi_K\left({\bf u}_k \mid {\bf x}_k \right) \mid \mid \nonumber \\
    &= \frac{1}{N} \sum_{k=1}^N \mid \mid \Sigma_{\sigma}^{-1} ({\bf u}_k + {\bf K}_j {\bf x}_k) {\bf x}_k^T \mid \mid = \frac{1}{N} \sum_{k=1}^N \mid \mid \Sigma_{\sigma}^{-1}\sigma_k {\bf x}_k^T \mid \mid \nonumber \\
    & \le ||\Sigma_{\sigma}^{-1}|| \frac{1}{N} \sum_{k=1}^N ||\sigma_k|| ||{\bf x}_k|| \nonumber \\
&\text{[Adding and subtracting $\mathbb{E} \left[ ||\sigma|| ||{\bf x}|| \right] = \left(tr(\Sigma_{\sigma}) tr(\Sigma_{K_j}) \right)^{1/2}$]} \nonumber \\
    & \le ||\Sigma_{\sigma}^{-1}|| \left[ \left(tr(\Sigma_{\sigma}) tr(\Sigma_{K_j}) \right)^{1/2} \right.  \nonumber \\
    & \left. + \frac{1}{N} \sum_{k=1}^N ||\sigma_k|| ||{\bf x}_k|| - \left(tr(\Sigma_{\sigma}) tr(\Sigma_{K_j}) \right)^{1/2} \right] \nonumber \\
    \le & ||\Sigma_{\sigma}^{-1}|| \left[ \left(tr(\Sigma_{\sigma}) tr(\Sigma_{K_j}) \right)^{1/2} + \epsilon_1 \right], \text{\wpp  } 1 - \delta_1,
    \label{eqn:grad_log_pi2}
\end{align}
where $\epsilon_1 > 0$ is a constant and $\delta_1 = \frac{tr(\Sigma_{\sigma}) tr(\Sigma_{K_j})}{N \epsilon_1}$. Using (\ref{eqn:grad_log_pi2}) in (\ref{eqn:grad_diff_term1}), we can write
\begin{align}
    T_1&\le \left(\frac{\epsilon}{2} + M_Q \right) ||\Sigma_{\sigma}^{-1}|| \left[ \left(tr(\Sigma_{\sigma}) tr(\Sigma_{K_j}) \right)^{1/2} + \epsilon_1 \right], \nonumber \\
    & \text{\wpp\  } 1 - \delta_c - \delta_1. \label{eqn:grad_diff_term1_final}
\end{align}
Next we derive the upper bound for the second term $T_2$ on the right-hand side of (\ref{eqn:grad_diff2}). Then we apply the matrix Bernstein inequality \cite{tropp2015introduction}, and assume there exists a constant $Q_{m} \ge 0$, such that $\mid \mid X_Q \mid \mid \le Q_{m}$ for all ${\bf x}, {\bf u}$ almost surely, where $X_Q ={Q}\left( {\bf x}_k, {\bf u}_k \right) \nabla_K \log \pi_K\left({\bf u}_k \mid {\bf x}_k \right) - \mathbb{E}[{Q}\left( {\bf x}, {\bf u} \right) \nabla_K \log \pi_K\left({\bf u} \mid {\bf x} \right)]$. Since the CL-system remains stable all the time, such an assumption is reasonable. Therefore, we can write
\begin{align}
    T_2 \le \epsilon_2, \text{\wpp\ } 1 - \delta_2, \label{eqn:grad_diff_term2_final}
\end{align}
where $\epsilon_2 > 0$ is a constant and $\delta_2 = (n+p) \exp\left(\frac{-N \epsilon_2^2/2}{\sigma_Q^2 + Q_m \epsilon_2/3} \right)$, and $\sigma_Q^2 = \max\left\{||\mathbb{E}[X_Q,X_Q^T]||, ||\mathbb{E}[X_Q^T,X_Q], || \right\} $. Using (\ref{eqn:grad_diff_term1_final}) and (\ref{eqn:grad_diff_term2_final}) in (\ref{eqn:grad_diff2}), we can write
\begin{align}
    &||\hat \nabla_K \mathcal{L}({\bf K}_j, \lambda) - \nabla_K \mathcal{L}({\bf K}_j, \lambda)|| \le \\
    & \left(\frac{\epsilon}{2} + M_Q \right) ||\Sigma_{\sigma}^{-1}|| \left[ \left(tr(\Sigma_{\sigma}) tr(\Sigma_{K_j}) \right)^{1/2} + \epsilon_1 \right] + \epsilon_2, \nonumber \\
    & \text{\wpp\ } 1 - \delta_c - \delta_1 - \delta_2.
    \label{eqn:grad_diff_final}
\end{align}
Finally, using (\ref{eqn:grad_diff_final}) in (\ref{eqn:K_diff2}), we can write
\begin{align}
    &||{\bf K}_{j+1} - {\bf K}^{'}_{j+1}|| \le \alpha_j ||{\Sigma}_{{\bf K}_j}^{-1}|| \left[ \left(\frac{\epsilon}{2} + M_Q \right) ||\Sigma_{\sigma}^{-1}|| \right. \nonumber \\
    & \left. \left(tr(\Sigma_{\sigma}) tr(\Sigma_{K_j}) \right)^{1/2} + \epsilon_1 \right] + \epsilon_2 + \epsilon_\Sigma, \text{\wpp\ } 1 - \delta_c - \delta_1 - \delta_2. 
    \label{eqn:K_diff_final}
\end{align}
Setting the right-hand side of the above inequality to $\epsilon_K$ and $\delta_K = \delta_c + \delta_1 + \delta_2$, we complete the proof of Lemma~\ref{lemma:K_diff}.
\end{proof}

Now using the L-smoothness property of the Lagrangian function $\mathcal{L}({\bf K}, \lambda)$ (Lemma~\ref{lemma:Lsmooth}), we can write 
\begin{align}
    &\mid \mathcal{L}({\bf K}_{j+1}, \lambda) - \mathcal{L}({\bf K}^{'}_{j+1}, \lambda) \mid \nonumber \\
    &\le \mid \text{tr}\left(\nabla_{\bf K} \mathcal{L}({\bf K}^{'}_{j+1}, \lambda)\right) \mid \epsilon_K + \frac{L}{2}r_K \epsilon_K^2, \text{\wpp\ } 1 - \delta_K. \nonumber \\
    &  \le \epsilon_L \text{ \wpp\ } 1 - \delta_K. \label{eqn:L_smooth_K_diff2}
\end{align}
Here, $r_K$ is the rank of the matrix $({\bf K}_{j+1} - {\bf K}^{'}_{j+1})$. 
Using (\ref{eqn:L_smooth_K_diff2}) in Lemma~\ref{lemma:conv_rate}, we can write
\begin{align}
    \mathcal{L}({\bf K}'_{j+1}, \lambda) - \mathcal{L}({\bf K}_j, \lambda) &\le  -(1-\beta)\left(\mathcal{L}({\bf K}_{j}, \lambda) - \mathcal{L}({\bf K}^*, \lambda) \right) \nonumber \\
    \mathcal{L}({\bf K}_{j+1}, \lambda) - \mathcal{L}({\bf K}_{j}, \lambda) &\le \epsilon_L \label{eqn:monotonicity_final} \\
    &-(1-\beta)\left(\mathcal{L}({\bf K}_{j}, \lambda) - \mathcal{L}({\bf K}^*, \lambda) \right)
    \nonumber
\end{align}
We assume there exists a sufficiently small $\epsilon_L >0$ such that
\begin{equation}
 \epsilon_L < (1-\beta)\left(\mathcal{L}({\bf K}_{j}, \lambda) - \mathcal{L}({\bf K}^*, \lambda) \right).
 \label{eqn:epsilon_L_bound}
\end{equation}
  Since $\epsilon_L$ can be made arbitrarily small by choosing sufficiently large actor iterations say $N_a$ and $M$ and step size $\alpha$ sufficiently small, such an assumption is reasonable. Therefore, from (\ref{eqn:monotonicity_final}), we can write
\begin{align}
    &\mathcal{L}({\bf K}_{j+1}, \lambda) - \mathcal{L}({\bf K}_{j}, \lambda) \le 0, \text{\wpp\ } 1 - \delta_K.
    \label{eqn:monotonicity_final3}
\end{align}
To find an upper bound on $\alpha_j$, such that the above inequality holds, we write the dependencies of $\epsilon_L$ on $\alpha_j$ using (\ref{eqn:L_smooth_K_diff2}) as 
\begin{equation}
\begin{aligned}
    \epsilon_L(\alpha_j) = c_1 \alpha_j + c_2 \alpha_j^2,
    \text{since } \epsilon_K = \mathcal{O}(\alpha_j) (\ref{eqn:K_diff_final}).
\end{aligned}
\end{equation}
Here, $c_1$ and $c_2$ depend on the problem parameters and state correlation matrix for the policy ${\bf K}_j$. Therefore, the learning rate should satisfy the following condition, 
\begin{align}
&0 < \alpha_j < \min\left(\frac{2}{LC_\Sigma}, \alpha_j^*\right) \text{, where} \nonumber \\
&\alpha_j^* = \sup \{\alpha_j > 0: c_1 \alpha_j + c_2 \alpha_j^2 < (1-\beta)\left(\mathcal{L}({\bf K}_{j}, \lambda) \right. \nonumber \\
 &- \left. \mathcal{L}({\bf K}^*, \lambda) \right) \}.
\end{align}

The term $\frac{2}{LC_\Sigma}$ is appearing from the requirement of the step size for the convergence of the policy gradient method with true natural policy gradient, see Lemma~\ref{lemma:conv_rate}.
Furthermore, from Lemma~\ref{lemma:conv_rate}, (\ref{eqn:L_smooth_K_diff2}), and (\ref{eqn:epsilon_L_bound}), we can write
\begin{align}
    &\mathcal{L}({\bf K}_{j+1}, \lambda) - \mathcal{L}({\bf K}^*, \lambda) \le \epsilon_L +  \beta \left(\mathcal{L}({\bf K}_{j}, \lambda) - \mathcal{L}({\bf K}^*, \lambda) \right),  \nonumber \\
    &< \left(\mathcal{L}({\bf K}_{j}, \lambda) - \mathcal{L}({\bf K}^*, \lambda) \right), \text{\ }  \label{eqn:monotonicity_final4} \\
    &\le \beta_1 \left(\mathcal{L}({\bf K}_{j}, \lambda) - \mathcal{L}({\bf K}^*, \lambda) \right), \text{ } 0< \beta_1 <1 \text{, \wpp\ } 1 - \delta_K, \nonumber
\end{align}
Hence, we complete the proof.
    \section{Proof of Lemma \ref{lemma:duality}} \label{apdx:duality}
    \modified{3.4}{
    Consider ${\bf X} = \Sigma_{\bf K} \succ 0$ and ${\bf Y} = {\bf K}\Sigma_{\bf K}$, a bijection on $\mathcal{K}$ with ${\bf K} = {\bf Y}{\bf X}^{-1}$ and $({\bf A}-{\bf B}{\bf K})\Sigma_{\bf K} = {\bf A}{\bf X}-{\bf B}{\bf Y}$. The cost becomes $\langle {\bf Q},{\bf X}\rangle + \text{tr}({\bf R}{\bf Y}{\bf X}^{-1}{\bf Y}^T)$, which is jointly convex in ${\bf X}, {\bf Y}$, since the epigraph of the matrix-fractional term admits a linear matrix inequality (LMI) representation by the Schur complement. The Lyapunov equation (\ref{eqn:Sigma_K}) likewise becomes an LMI and, by (\ref{eqn:equiv_constraint}), the constraint is the linear inequality $\langle {\bf q}{\bf q}^T,{\bf X}\rangle \le \epsilon_\delta$. Therefore (\ref{eqn:opt_prob}) is equivalent to the convex semidefinite program
    \begin{align}
        \min_{{\bf X},{\bf Y},\Pi} \; & \langle {\bf Q},{\bf X}\rangle + \text{tr}(\Pi) \nonumber \\
        \text{s.t. } \; & {\bf X} \succ 0, \quad \langle {\bf q}{\bf q}^T,{\bf X}\rangle \le \epsilon_\delta, \nonumber \\
        & \begin{bmatrix} \Pi & {\bf R}^{1/2}{\bf Y} \\ \star & {\bf X}\end{bmatrix} \succeq 0, \nonumber \\
        & \begin{bmatrix} {\bf X}-\Sigma_{\bar w} & {\bf A}{\bf X}-{\bf B}{\bf Y} \\ \star & {\bf X}\end{bmatrix} \succeq 0, \label{eqn:sdp}
    \end{align}
    with the optimal policy recovered as ${\bf K}^{*} = {\bf Y}^{*}({\bf X}^{*})^{-1}$. The relaxation of the Lyapunov equality is tight at the optimum, since for a fixed ${\bf K}$ the set $\{{\bf X} : {\bf X} \succeq ({\bf A}-{\bf B}{\bf K}){\bf X}({\bf A}-{\bf B}{\bf K})^T + \Sigma_{\bar w}\}$ has $\Sigma_{\bf K}$ as its minimal element, while $\langle {\bf Q}+{\bf K}^T{\bf R}{\bf K},{\bf X}\rangle$ increases with ${\bf X}$. Hence the nonconvexity of (\ref{eqn:opt_prob}) in ${\bf K}$ is only benign. \newline 
    By Assumption~\ref{assump:slater}, the strictly feasible $\bar{\bf K}$ satisfies $\langle {\bf q}{\bf q}^T,\Sigma_{\bar{\bf K}}\rangle < \epsilon_\delta$, so Slater's condition holds for (\ref{eqn:sdp}). Since (\ref{eqn:dual_prob}) is the Lagrange dual of (\ref{eqn:sdp}) with respect to the constraint $\langle {\bf q}{\bf q}^T,{\bf X}\rangle \le \epsilon_\delta$, standard convex duality gives a zero duality gap, $D^* = J^*$, together with a finite optimal multiplier $\lambda^* \ge 0$ and complementary slackness $\lambda^*\left(J_c({\bf K}^*) - \epsilon_\delta\right) = 0$. \newline 
    The change of variables and the resulting convex reformulation follow the hidden-convexity framework developed for the unconstrained LQR in \cite{watanabe2025revisiting}. In contrast to that work, which dualises the Lyapunov equality constraint with a matrix-valued multiplier, we retain the Lyapunov constraint and dualise the scalar risk constraint, so that Slater's condition is required here; convexity is preserved because, by (\ref{eqn:equiv_constraint}), the risk constraint is linear in ${\bf X}$.}


\section{Parameters}
\label{apdx:params} 
\modified{1.6.1b}{
\begin{align*}
	&{\bf A} =\begin{bmatrix}
		1 & 0.5 & 0 & 0  \\
		0 & 1 & 0& 0 \\
		0 & 0 & 1 & 0.5 \\
		0 &0 & 0 & 1
	\end{bmatrix}, {\bf B} =\begin{bmatrix}
	0.125 & 0  \\
	0.5 & 0 \\
	0 & 0.125 \\
	0 & 0.5
\end{bmatrix},  \\
& \Sigma_w = \textbf{B}diag(80, 80, 0.01)\textbf{B}^T, \\
    &\alpha_j = 0.1, N_c = 250, L_c = 1000, N = 1000, c_\alpha = 10.
\end{align*}} \vspace{-4mm}

\section{Algorithms}
\label{apdx:algorithms}
\begin{algorithm}[h!]
    \caption{Chance constrained LQR (CLQR)}
    \label{algo:SDP}
    \begin{algorithmic}
        \State Apply SDP to solve the following optimisation problem and obtain the optimal controller as ${\bf K_{sdp}} ={\bf Y} {\bf X}^{-1}$.
        \begin{equation*}
            \begin{aligned}
                &\min_{{\bf{X}}, {\bf Y}, {\bf P}} \quad \text{Tr}({\bf QX}) + \text{Tr}({\bf P})\\
                &\text{s.t.} \quad \begin{bmatrix}
                    {\bf P} & ({\bf R}^{1/2} {\bf Y}) \\
                    ({\bf R}^{1/2} {\bf Y})^T & {\bf X}
                \end{bmatrix} \succeq 0,\\
                &\quad \begin{bmatrix}
                    {\bf X} - {\bf W} & {\bf A}{\bf X} + {\bf B}{\bf Y} \\
                    ({\bf A}{\bf X} + {\bf B}{\bf Y})^T & {\bf X}
                \end{bmatrix} \succeq 0,\\
                &\quad q^T{\bf X} q \leq \alpha \epsilon^2, \text{ where } \alpha = (nrm^{-1}(1-\delta))^{-2}.
            \end{aligned}
        \end{equation*}
        \State Here $nrm(\cdot)$ is the cumulative normal distribution function. ${\bf{X}} \in {\rm I\!R}^{n\times n} >0$, ${\bf{P}} \in {\rm I\!R}^{p\times p} \geq 0$, and ${\bf{Y}} \in {\rm I\!R}^{p\times n}$.
    \end{algorithmic}
\end{algorithm}

\begin{algorithm}[h!]
    \caption{Scenario-based chance-constraint MPC}
    \label{algo:MPC}
    \begin{algorithmic}
        \State Perform the following steps at each time step $t$:
    
        \State Measure the current state ${\bf x}_t$.
    
        \State Generate $S$ noise samples ${\bf w}_t^{(1)}, \cdots, {\bf w}_t^{(S)} \sim f_w({\bf w})$.
    
        \State Solve the following optimisation problem:
        \begin{equation*}
            \begin{aligned}
                &\min_{{\bf u}_{1|t},\cdots,{\bf u}_{T|t}} \sum_{s=1}^S \sum_{i=1}^T \left(f\left( {\bf x}_{i|t}^{(s)},{\bf u}_{i|t}\right) + \lambda \mathds{1}_{\left\{f_c\left({\bf x}_{i+1|t}^{(s)}\right) \ge \epsilon\right\}}   \right)\\
                &\text{s.t. (\ref{eqn:state_eqn}) is satisfied. }
            \end{aligned}
        \end{equation*}
    
        \State Apply the first control input ${\bf u}_{1|t}$ to the system. 
        \State [$f(\cdot|\cdot)$ is given in (\ref{eqn:fk}). ${\bf x}_{i|t}$ and ${\bf u}_{i|t}$ denote predictions and plans of the state and input variables made at time $t$, for $i$ steps into the future.]
    \end{algorithmic}
\end{algorithm}
\vspace{-2mm}


\bibliographystyle{IEEEtran}
\bibliography{IEEEabrv,Const_Control}

\end{document}